\documentclass[nohypdvips,reqno]{siamart251216}

\usepackage{fontspec}

\usepackage{polyglossia}

\setdefaultlanguage{english}

\usepackage{tikz}
\usetikzlibrary{positioning}
\usepackage{ulem}
\usepackage{csquotes}
\usepackage[a4paper,left=2cm,right=2cm]{geometry}
\usepackage{ dsfont }
\usepackage{ upgreek }
\usepackage{hyperref}
\usepackage{xcolor}
\hypersetup{
    colorlinks,
    linkcolor={black!50!black},
    citecolor={black!50!black},
    urlcolor={black!80!black}
}

\usepackage{amssymb} 
\usepackage{graphicx}
\usepackage{ mathrsfs }
\usepackage{mathtools}
\usepackage{cleveref}

\def\be#1#2\ee{\begin{equation}\label{eq:#1}#2\end{equation}}
\def\req#1{{\rm(\ref{eq:#1})}}

\newcommand*\diff{\mathop{}\!\mathrm{d}}

\newcommand*\dx{\diff x}
\newcommand\dxx{\diff \textit{\textbf{x}}}
\newcommand*\dy{\diff y}
\newcommand*\dyy{\diff \textit{\textbf{y}}}

\newcommand*\dr{\diff r}

\newcommand{\BF}{\boldsymbol{F}_{\{x_1,x_2\}}}

\newcommand*{\R}{\mathbb{R}}
\newcommand*{\N}{\mathbb{N}}
\newcommand*{\A}{\mathcal{A}}

\newcommand{\fexp}{\exp^\ast}

\newcommand{\bP}{\boldsymbol{\Psi}}
\newcommand{\bPH}{\boldsymbol{\Phi}}
\newcommand{\bro}{\boldsymbol{\rho}}
\newcommand{\bom}{\boldsymbol{\omega}}
\newcommand{\bF}{\boldsymbol{F}}
\newcommand{\xx}{\textit{\textbf{x}}}
\newcommand{\yy}{\textit{\textbf{y}}}
\newcommand{\zz}{\textit{\textbf{z}}}

\newtheorem{example}{Example}
\newtheorem{remark}{Remark}

\newcommand{\PP}{\ensuremath{\mathsf{P}}}
\newcommand{\KK}{\ensuremath{\mathsf{K}}}

\newcommand*\dP{\diff \PP}
\renewcommand{\R}{\mathbb{R}}

\newtheorem{assumption}{Assumption}

\newcommand{\fod}[1]{\omega_1^{(1+#1)}}
\newcommand{\secod}[1]{\omega_2^{(2+#1)}}
\newcommand{\rhot}[1]{\rho^{(#1)}_T}
\newcommand{\pmf}[1]{u_{\textnormal{PMF}}^{(#1)}}
\newcommand{\mayerpmf}[1]{f_{\textnormal{PMF}}^{(#1)}}

\title{An inversion formula for the 2-body interaction given the correlation functions}
\author{F.~Frommer\thanks{Institut f\"ur Mathematik, Johannes
    Gutenberg-Universit\"at Mainz, 55099 Mainz, Germany
    ({\tt fafromme@uni-mainz.de})} \and
    T.~Kuna\thanks{Dipartimento di Ingegneria e scienze dell'informazione e matematica, 
        Università degli Studi dell'Aquila,
        Palazzo Camponeschi, piazza Santa Margherita 2, 67100 L'Aquila, Italy
    {\tt     ({\tt tobias.kuna@univaq.it})}}, \and 
        D.~Tsagkarogiannis\thanks{Dipartimento di Ingegneria e scienze dell'informazione e matematica, 
        Università degli Studi dell'Aquila,
        Palazzo Camponeschi, piazza Santa Margherita 2, 67100 L'Aquila, Italy ({\tt dimitrios.tsagkarogiannis@univaq.it})}. 
	}

\begin{document}
%\newfontfamily{\hebrewfonttt}[Script=Hebrew]{Miriam Mono CLM}
\maketitle
\sloppy
\begin{abstract}
Given a classical gas described by the truncated correlation functions of all orders, 
we prove convergence of an expansion of the pair interaction part of the (unknown) potential in terms of the truncated correlation functions of all orders, at infinite volume.  
\end{abstract}

\section{Introduction}

The main goal of statistical physics is to derive
thermodynamic quantities from microscopic models of interacting particles based on first principles. 
At equilibrium, isolated systems are described via the Gibbs measures: a possible configuration of particles is more likely when it is energetically favorable. However, when working with actual data, these interactions are impossible to measure. What instead can be either measured experimentally or read out of simulation data are the so called correlation functions $\rho^{(n)}$ of the corresponding Gibbs measure. This gives rise to the inverse question: given the correlation functions, can one determine the $n$-body interactions that produce it? This is of fundamental importance, since it may permit designing systems with special properties (as induced by the correlations) by properly tuning the atomistic parameters of the system such as the interaction potential. From a theoretical point of view, i.e., concerning the existence and uniqueness of a solution, this problem has been investigated for the cases $n=1,2$ in, e.g.~\cite{Chayes83},\cite{Koral07},\cite{FFMH19}.
In practice, this is relevant from a computational point of view as well: it provides a strategy for the design of efficient computational methods by suggesting effective models at a coarser and computationally more tractable scale by matching the correlation functions, cf.~\cite{Noid13b,Peter09}.

The main idea can be briefly described as follows: for a microscopic system with a pair interaction, one can write the expansion for the pair correlation as a power series in terms of the density of the system.
By keeping the first term of this expansion, one relates the pair correlation function to the pair potential with some error due to the higher order terms. Suppose now that for a microscopic system with an unknown interaction potential (possibly multi-body), one has the pair correlation function by some experiment or simulation. Using the previous formula, we can compute an approximate (due to the neglected higher order corrections) pair interaction potential. With that, we can simulate a new correlation function, propagating the original error. The question is how to control and hopefully improve this error. One suggestion is to iterate this procedure and eventually find a good candidate for an effective pair interaction potential. By construction, it will approximate the original pair-correlation function well, but it might not fit well with pressure or other thermodynamic quantities. This method is the well-known {\it Iterative Boltzmann inversion} \cite{Muellerplathe02,Muellerplathe03}, with many variants {such as the inverse Monte-Carlo approach \cite{Lyubartsev95,Lyubartsev04}} and attempts to also match the thermodynamic behavior, e.g.~\cite{Clark13,McCarty12}. \\

Alternatively, another suggestion is, instead of iterating, to invert and express the pair potential as a power series of the pair correlation function and compute the next order in the expansion. This was established in \cite{Gstell64} and in \cite{Morita61}, but to our knowledge, the proof of the radius of convergence is still pending. 

In the present paper, instead of inverting the relation between the pair correlation/interaction function, we consider the correlations of all orders and relate them to the potentials of all orders. The main idea is to exploit two  well known relations; for each bounded subset of the full space one can define the Janossy densities of the Gibbs measure, that is the marginal obtained only observing the point processes in $\Lambda$. The first relation relates the Janossy density to the correlation functions of the point processes via an explicit formula. The second relation is  the {GNZ}-equation itself. Comparing the two formulas, we obtain an expression of the pair interaction as a power series with respect to correlation functions of all orders in the limit $\Lambda \uparrow \R^d$. This idea goes back to an approach by Nettleton and Green for a grand canonical ensemble, cf.~\cite{RNMG57}. The case of the one body interaction was addressed in a previous paper by one of the authors \cite{Frommer24}. Here, we consider the two body interaction leading to more challenging combinatorial considerations. As expected, the next order correction for the two body interaction is given by a combination of correlation functions up to third order. This connects nicely with various attempts to compute three-body interactions either directly \cite{Dimitris} or combining ideas from stochastic geometry and statistics \cite{Gottschalk,StoyanStoyan} and opens the discussion for future investigation on finding the most efficient method.

The structure of the paper is as follows: after introducing the various relevant quantities, in Section~\ref{sec2} we recall the Ruelle calculus which provides a natural framework for our calculations. The main results are presented in Section~\ref{sec3} together with their proofs. We conclude with some examples in Section~\ref{sec4}.

\section{Setting}\label{sec2}
A point process $\PP$ is a probability measure on the set of configurations 
\begin{align*}
    \Gamma = \{\eta \subset \R^d \mid N(\eta \cap \Delta )< \infty \text{ for all } \Delta \subset \R^d \text{ bounded}\}
\end{align*}
equipped with the $\sigma$-algebra $\mathcal{F}:=\sigma\left(\{N(\cdot\cap \Delta)=m\}\mid m\in\N_0, \,\,\Delta\subset \R^d \text{ bounded}\right)$. Here $N(\eta)=\#\eta$ is the number of elements of $\eta\subset \R^d$. We denote by $\Gamma_0= \{\gamma \in \Gamma \mid N(\gamma)< \infty\}$ the set of finite configurations. Such a point process is a Gibbs measure for some chemical potential $\mu \colon \R^d \to \R$ and a Hamiltonian $H\colon \Gamma\to \R\cup\{\infty\}$, we write $\PP$ is a $(\mu,H)$-Gibbs measure, if it is supported on a suitable set of \emph{tempered} configurations and for corresponding conditions on the Hamiltonian 
it satisfies the \emph{GNZ-equation}, see \cite{Zessin79}, namely if, for every non-negative function $G\colon\Gamma \to [0,\infty]$, there holds
\begin{align}\label{eq:GNZ}
    \int_{\Gamma} \sum_{x\in \eta }G(x,\eta)\dP(\eta) = 
    \int_{\R^d}\int_{\Gamma}
    G(x,\{x\}\cup\eta){e^{\mu(x)-W(\{x\}\mid \eta)}}\dP(\eta)\dx.
\end{align}
% \begin{align*}
%     \int_{\Gamma} G(\eta)\dP(\eta) = 
%     \int_{\Gamma_{\Lambda^c}}\sum_{n=0}^\infty\frac{e^{\sum_{i=1}^n \mu(x_i)}}{n!}\int_{\Lambda^n} 
%     G(\{\xx_n\}\cup\eta){e^{-H(\{\xx_n\})-W(\{\xx_n\}\mid \eta)}}\dxx_n\dP(\eta)
% \end{align*}
% where $\Lambda^c=\R^d\backslash\Lambda$, $\Gamma_{\Lambda^c} = \{\eta\in\Gamma \mid \eta \subset \Lambda^c\}$ and 
Here, $W(\{x\}\mid \eta)$ is the extra energy cost associated with adding the point $x$ to the configuration $\eta$. This interaction energy is defined for finite configurations via the relation
\begin{align*}
    H(\eta'\cup \eta)=
    H(\eta')+W(\eta'\mid \eta)+H(\eta), \qquad \eta,\eta'\in \Gamma_0 
\end{align*}
and it can be extended to infinite configurations in a reasonable way, see \cite{Ruelle69} and \cite{Vasseur20}. However, the specific definition of $W$ is not important for this work, and thus details are left out.
%%:= $\lim_{\Delta \nearrow \R^d}W(\{\xx_n\}\mid \eta\cap \Delta)$
%%and where $\eta_{\Lambda^c}=\eta\cap\Lambda^c$
%\begin{align*}
%    W_\Delta(\{\xx_n\}\mid \eta):=
%    H(\{\xx_n\}\cup (\eta\cap\Delta))
%    -H(\{\xx_n\})-H(\eta\cap \Delta), 
%    \qquad \xx_n\in \Lambda^n, \,\eta\in\Gamma_{\Lambda^c}
%\end{align*}
%with the convention that we set $W_\Delta(\{\xx_n\}\mid \eta)=+\infty $ if $H(\{\xx_n\}\cup (\eta\cap\Delta))=+\infty$. We make the assumption on $H$ that the limit $\lim_{\Delta\nearrow \R^d}W_\Delta(\{\xx_n\}\mid \eta)$ exists for every $\eta \in \Gamma_{\Lambda^c}$ (here $\Delta\nearrow \R^d$ means that we take an increasing sequence of bounded sets $\Delta$ such that eventually every bounded set of $\R^d$ is contained in them) and define $W(\{\xx_n\}\mid \eta)=\lim_{\Delta\nearrow \R^d}W_\Delta(\{\xx_n\}\mid \eta)$.}
%is the interaction between the part of the particles inside $\Lambda$ and those on the outside,  We write $\PP$ is a $(\mu,H)$-Gibbs measure.
%By abuse of notation we identify $\xx_n$ as the element $\{x_1,\dots,x_n\}\in\Gamma$, where ignore the cases that two $x_i$ and $x_j$ coincide as this is a submanifold of $(\R^d)^n$ of Lebesgue measure zero.
It is also assumed that the Hamiltonian $H$ is \emph{translation invariant}.  {Note that the one body interaction is treated separately in the chemical potential and it is thus assumed that $H(\{x\})=0$.} The last assumption made on $H$ is \emph{stability},
i.e., there is a constant $B>0$ such that 
\begin{align*}
    H(\xx_n) \geq -Bn,
\end{align*}
where $\xx_n:=\{x_1,\ldots,x_n\}\subset\mathbb R^d$.
Under some additional technical assumptions, as alluded previously, there exists at least one translation invariant $(\mu,H)$-Gibbs measure $\PP$ associated with these parameters and for the rest of this section we work in this framework.
From the GNZ-equation \req{GNZ} it follows that the so-called \emph{correlation functions} $(\rho^{(n)})_{n\geq 1}$ of such a Gibbs measure $\PP$ exist and are given by
\begin{align}\label{eq:CorrGibbs}
    \rho^{(n)}(\xx_n) =  \int_{\Gamma} \exp\Big(\sum_{i=1}^n \mu(x_i)-H(\{\xx_n\})-W(\{\xx_n\}\mid \eta)\Big) \dP(\eta).
\end{align}
Note that when $\PP$ is translationally invariant, the density of $\PP$ is constant, i.e.,~$\rho^{(1)}\equiv \rho$. If these correlation functions satisfy that there is a $\xi>0 $ such that 
\begin{align}\label{eq:Rbound}
    0 \leq \rho^{(n)}(\xx_n) \leq \xi^n  \quad\text{ for all } \xx_n \subset \R^{d}, \quad n\geq 1 \tag{R-bound},
\end{align}
it is said that $\PP$ satisfies a \emph{Ruelle bound} and this will be our key assumption. Furthermore, in this case, the finite volume densities of $\PP$, also called \emph{Janossy densities}, for some bounded set $\Lambda\subset \R^d$ can be written in terms of the correlation functions as
\begin{align}\label{eq:Janossycorr}
     j_\Lambda^{(n)}(\xx_n) = \sum_{k=0}^\infty\frac{(-1)^k}{k!} \int_{\Lambda^k}\rho^{(n+k)}(\xx_n,\yy_k)\dyy_k,
\end{align}
see, e.g. \cite{jansen18} (with the convention that $\rho^{(0)}:=1$). More explicitly, plugging \req{CorrGibbs} into \req{Janossycorr} shows that 
\begin{align}\label{eq:Ruelle_Formula}
    j_\Lambda^{(n)}(\xx_n) =  \int_{\Gamma_{\Lambda^c}} \exp\Big(\sum_{i=1}^n \mu(x_i)-H(\xx_n)-W(\xx_n\mid \eta)\Big) \dP(\eta),
\end{align}
where $\Gamma_{\Lambda^c} = \{\eta\in\Gamma \mid \eta \subset \Lambda^c\}$. In other words, there holds
\begin{align}\label{eq:JanossyGibbs}
    \int_{\Gamma_{\Lambda^c}} \exp\Big(\sum_{i=1}^n \mu(x_i)-H(\xx_n)-W(\xx_n\mid \eta)\Big) \dP(\eta)
    &=j_\Lambda^{(n)}(\xx_n) 
    \\ & =\sum_{k=0}^\infty\frac{(-1)^k}{k!} \int_{\Lambda^k}\rho^{(n+k)}(\xx_n,\yy_k)\dyy_k.
    \nonumber 
\end{align}
\\
The idea of this work is to write the Hamiltonian $H$ as a power series of
the correlation functions $(\rho^{(m)})_{m\geq 1}$. 
To do so, we take the logarithm of the above equation and for $n=2$ we obtain:
\begin{align}\label{eq:j2hamiltonian}
    \log j_\Lambda^{(2)}(\xx_2)= \mu(x_1)+\mu(x_2)-H(\xx_2) + \log \int_{\Gamma_{\Lambda^c}} \exp\Big(-W(\xx_2\mid \eta)\Big) \dP(\eta). 
\end{align}
Observe that for finite range Hamiltonians, by choosing $\Lambda$ large enough, the contribution $W(\xx_2\mid \eta)$ is zero. For general Hamiltonians, one may reasonably expect that:
\begin{align*}
    \log \int_{\Gamma_{\Lambda^c}} \exp\Big(-W(\xx_2\mid \eta)\Big) \dP(\eta)
    =
    \log \int_{\Gamma_{\Lambda^c}}  \dP(\eta)+
    \varepsilon(\Lambda),
\end{align*}
where $\varepsilon(\Lambda)\to 0 $ as $\Lambda \nearrow\R^d$.
Note that the first term on the right hand side is equal to $\log j_\Lambda^{(0)}$ {by \eqref{eq:Ruelle_Formula}
}
and
hence we expect that
% , 
\begin{align}\label{eq:j2overj2limit}
    \log j_\Lambda^{(2)}(\xx_2)-\log j_\Lambda^{(0)} \to \mu(x_1)+\mu(x_2)-H(\xx_2)
\end{align} as $\Lambda\nearrow\R^d$. Thus, in order to express the Hamiltonian as a power series in $(\rho^{(m)})_{m\geq 1}$ we need to expand the $\log j_\Lambda^{(2)}(\xx_2)$.

\begin{remark}
It may be more intuitive to consider the same question in finite volume. The \emph{grand canonical ensemble} in a finite volume $\tilde{\Lambda}$ is given by densities
\begin{equation}
    j^{(n, \tilde{\Lambda})}(\xx_n)  = \frac{1}{\Xi_{\tilde{\Lambda}}} e^{\sum_{i=1}^n \mu(x_i) - H(\{\xx_n\})},
\end{equation} 
where the normalization constant $\Xi_\Lambda$ is the so-called \emph{grand canonical partition function}. What is the relation with the Janossy densities?
If we denote the associated grand canonical (finite volume) Gibbs measure by $P_{\tilde{\Lambda}}$ then \eqref{eq:Ruelle_Formula} holds with $P$ replaced by $P_{\tilde{\Lambda}}$ and $\Gamma_{\Lambda^c}$ replaced by $\Gamma_{\tilde{\Lambda} \cap \Lambda^c}$, that is the Janossy denisities of $P_{\tilde{\Lambda}}$ are
\begin{equation}
j^{(n, \tilde{\Lambda})}_\Lambda(\xx_n) = \sum_{k=0}^\infty \frac{1}{k!}\int_{(\tilde{\Lambda} \cap \Lambda^c)^k} e^{\sum_{i=1}^n \mu(x_i) -H(\xx_n) -W(\xx_n ,\yy_k)} \frac{1}{\Xi_{\tilde{\Lambda}}} e^{-\sum_{j=1}^k\mu(y_j) - H(\yy_k)}d\yy_k
\end{equation}
so $j^{(n,\tilde{\Lambda})}= \lim_{\Lambda \uparrow \tilde{\Lambda}} j^{(n,\tilde{\Lambda})}_\Lambda$. In the latter case \eqref{eq:j2overj2limit} becomes just 
\begin{equation}\label{eq:gcjanconv}
\frac{j^{(n,\tilde{\Lambda})}(\xx_n)}{j^{(0, \tilde{\Lambda})}} = e^{\sum_{i=1}^n \mu(x_i) -H(\xx_n)}; 
\end{equation}
even without limits. Our consideration above allows us to consider also the system directly in the infinite volume as well.
\end{remark}

\subsection{Ruelle calculus}
As mentioned in the previous section, the goal is to write the Janossy densities as an exponential. This fits nicely in the framework of \emph{Ruelle calculus} introduced in \cite{Ruelle69}, see also \cite{KunaPhD}: for any family of symmetric functions $(F_n)_{n\geq 0}$ with $F_0\in\R$ and $F_n\colon (\R^d)^n\to \R$ we can associate a function $\boldsymbol{F}\colon \Gamma_0\to \R$ by $\boldsymbol{F}(\{x_1,\dots,x_n\}):=F^{(n)}(x_1,\ldots,x_n)$ and write $\boldsymbol{F}=(F^{(n)})_{n\geq 0}$. We will also write $F^{(n)}(\xx_n)$ if there is no risk of confusion. We introduce the set $\mathcal{A}$ of all measurable functions on $\Gamma_0$ and the subset
\begin{align}
    \mathcal{A}_0:= \{\Phi\in\mathcal{A}\mid \Phi(\emptyset)=0 \}.
\end{align}
We introduce the operation $\ast\colon \A\times \A \to \A$ by
\begin{align}
    (\bP\ast \boldsymbol{\Phi})(\eta) := \sum_{\gamma\subset\eta} \bP(\gamma)\boldsymbol{\Phi}(\eta\backslash\gamma).
\end{align}
It is well-known that $\A$ together with $\ast$ forms a commutative algebra with unit element $\mathds{1}\in\A$ given by
\begin{align*}
    \mathds{1}(\eta) = \begin{dcases}
        1 ,\qquad \eta = \emptyset \\
        0, \qquad  \eta\neq \emptyset.
    \end{dcases}
\end{align*}
We also introduce the mapping $ \fexp\colon \A_0\to \mathds{1}+ \A_0$ by
\begin{align*}
    \fexp \bP:=\mathds{1}+ \sum_{n=1}^\infty \frac{1}{n!}\bP^{\ast n},
\end{align*}
where $\bP^{\ast 1}=\bP$ and
\begin{align*}
    \bP^{\ast n}(\eta) = 
    (\bP\ast\bP^{\ast (n-1)})(\eta).
\end{align*}
In particular, there holds $(\fexp \bP)^{(1)}(x)=\bP^{(1)}(x)$ and for $n\geq 2$
\begin{align*}
    (\fexp \bP)^{(n)}(\xx_n) = \sum_{k=1}^{n } \frac{1}{k!}\sum_{\pi \in\Pi_k(n)}\prod_{i=1}^k \bP^{( \kappa_i)}(\xx_{\pi_i}),
\end{align*}
where $\Pi_k({n})$ is the set of ordered partitions of $\{1,\dots ,n\}$ into $k$ sets $(\pi_i)_{1\leq i\leq k}$, $\kappa_i$ is the number of elements of $\pi_i$ and $\xx_{\pi_i}=(x_{\pi_i(1)},\dots, x_{\pi_i(\kappa_i)})$.
The mapping $\fexp$ satisfies the exponent rules, i.e.,
\begin{align}\label{eq:powerrule}
    \fexp \bP \ast \fexp \boldsymbol{\Phi}
    =\fexp (\bP+\boldsymbol{\Phi}).
\end{align}
\begin{example}
The well-known \emph{truncated correlation functions} are defined by the recursion $\rho_T^{(1)}= \rho^{(1)}=\rho$ and for $n\geq 2$ 
\begin{align}\label{eq:clusterfunctions}
    \rho^{(n)}_T(\xx_n)=\rho^{(n)}(\xx_n)- \sum_{k=2}^n \frac{1}{k!}\sum_{\pi \in\Pi_k({n})}\prod_{i=1}^k \rho_T^{(\kappa_i)}(\xx_{\pi_i}),%(-1)^{|\pi_i|}
\end{align}
where $\Pi_k({n})$ is as above. With the conventions $\rho^{(0)}=1$ and $\rho^{(0)}_T=0$. Using Ruelle's formulation we have that $\boldsymbol{\rho}=(\rho^{(n)})_{n\geq 0}\in\A$ and $\boldsymbol{\rho}_T=(\rho^{(n)}_T)_{n\geq 0} \in \A_0$. Thus, $\bro = \fexp \boldsymbol{\rho}_T$.
\end{example}

It is easy to show that for functions related by $\bP=\fexp\bF$ the following holds (see, e.g.~\cite{Rebenko14}):

\medskip
\begin{lemma}[Exponential representation]\label{thm:reb}
Suppose $\boldsymbol{\Phi}=(\Phi^{(n)})_{n\geq 0}\in \A_0 $ such that for all $n\geq 0$ and $\Lambda\subset\R^d$ bounded there holds 
\begin{align}\label{eq:Fbound}
    \int_{\Lambda^n}|\Phi^{(n+1)}(x,\xx_n)|\dxx_n \leq %|\Lambda| 
    n! C c^n,
\end{align}
for some $C>0$ and  $0<c<1/2$. Then for $\bP = \fexp \boldsymbol{\Phi}$
% Then the function $\Phi:\Gamma_0\to \R$ defined by
% \begin{align}\label{eq:partitionrep}
%     \Phi(\eta) := \sum_{k=1}^{ N(\eta) } \sum_{\pi \in\Pi_k(\eta)}\prod_{i=1}^k F_{\kappa_i}({\pi_i})
% \end{align}
% satisfies 
we have
\begin{align}\label{eq:exprep}
    \sum_{n=0}^\infty \frac{1}{n!} \int_{\Lambda^n} \Psi^{(n)}(\xx_n)\dxx_n = \exp\left(\sum_{n=1}^\infty \frac{1}{n!} \int_{\Lambda^n}\Phi^{(n)}(\xx_n) \dxx_n\right).
\end{align}
% for $0<c<1/2$.
\end{lemma}

From this lemma, supposing that the truncated correlation functions $(\rho_T^{(n)})_{n\geq 1}$ satisfy a bound of the type \req{Fbound}, then from \req{Janossycorr} for $n=0$ it follows that 
\begin{align}\label{eq:j0rhot}
    j^{(0)}_\Lambda = \exp \left(
    \sum_{k=1}^\infty\frac{(-1)^k}{k!} \int_{\Lambda^k}\rho^{(k)}_T(\yy_k)\dyy_k    
    \right),
\end{align}
where the series is absolutely convergent.
Under assumption \eqref{eq:Rbound}, a bound of the type \req{Fbound} holds, see \cite{Ruelle69}.
Lastly, for $\gamma\in \Gamma_0$ we introduce the operator $D_\gamma\colon \A\to \A$ by 
\begin{align*}
    (D_\gamma \bP) (\eta) :=\begin{dcases}
        \bP(\eta\cup\gamma), \qquad \eta\cap\gamma =\emptyset;\\
        0 ,\qquad \qquad\qquad \text{ else.}
    \end{dcases}
\end{align*}
Note that we have $D_{\gamma}=D_{\{x\}}D_{\gamma\backslash\{x\}}$ and it holds, cf.~\cite{KunaPhD},
\begin{align}\label{eq:prodrule}
    D_{\{x \}}(\bP\ast\bPH) =
    (D_{\{ x \} }\bP)\ast\bPH+\bP\ast (D_{\{ x \} }\bPH)
\end{align}
and
\begin{align}\label{eq:chainrule}
    D_{\{x\}}\left(\fexp \bP\right) = \fexp\bP\ast D_{\{x\}}\bP.
\end{align}
% \begin{align*}
%     j^{(2)}_\Lambda(\xx_2)= \sum_{k=0}^\infty\frac{(-1)^k}{k!} \int_{\Lambda^k}\rho^{(2+k)}(\xx_2,\yy_k)\dyy_k =  \rho^{(2)}(\xx_2)\left(1+\sum_{k=1}^\infty\frac{(-1)^k}{k!} \int_{\Lambda^k}\frac{\rho^{(2+k)}(\xx_2,\yy_k)}{\rho^{(2)}(\xx_2)}\dyy_k\right),
% \end{align*}
% and we thus (formally) have
In light of \eqref{eq:j2overj2limit}, in order to obtain an expression for the Hamiltonian, confer also the heuristic \eqref{eq:gcjanconv}, 
we need to consider $\log j_{\Lambda}^{(n)}$. To get a good expansion in terms of $\bro$ we want to expand the logarithm near $1$,
but $j_{\Lambda}^{(n)}$ may become very small, for example, when some particles in $\xx_n$ are close, but $\rho^{(n)}$ will show the same behaviour. 
Comparing  \eqref{eq:CorrGibbs} and \eqref{eq:Ruelle_Formula} one may convince oneself that $j_\Lambda^{(n)}/\rho^{(n)}$ should stay away from zero. This fraction also has a probabilistic interpretation: $j^{(m+n)}(\zz_m,\xx_n)/\rho^{(n)}(\xx_n)$ is the probability finding particles at $\zz_m, \xx_n$ in the volume $\Lambda$ (ignoring the outside) conditioned on having particles at $\xx_n$. The truncated correlation functions associated to this point process  are the family $(F_{n}^{(k)})_{n\geq 1}$ given by the relation  
\begin{align}\label{eq:jnexp}
\frac{j^{(n+m)}_\Lambda(\zz_m,\xx_n)}{\rho^{(n)}(\xx_n)}= \exp\left(\sum_{k=1}^\infty\frac{(-1)^k}{k!} \int_{\Lambda^k}F_{n}^{(k)}(\xx_n;\zz_m, \yy_k)\dyy_k\right).
\end{align}
or equivalently
\begin{align}\label{eq:rhonexp}
\frac{\rho^{(n+m)}_\Lambda(\yy_m,\xx_n)}{\rho^{(n)}(\xx_n)}= \exp^*\left(F_{n}(\xx_n;\cdot)\right)(\yy_n) .
\end{align}
Furthermore, using \req{j0rhot} one gets
\begin{align*}
    \frac{j^{(n)}_\Lambda(\xx_n)}{j^{(0)}_\Lambda}=\rho^{(n)}(\xx_n) \exp\left(\sum_{k=1}^\infty\frac{(-1)^k}{k!} \int_{\Lambda^k}\Big(F_{n}^{(k)}(\xx_n;\yy_k)-\rhot{k}(\yy_k)\Big)\dyy_k\right).
\end{align*}
This will be the context of our main result in the next section in the case  $n=2$:

\begin{align}\label{eq:j2exp}
    j^{(2)}_\Lambda(\xx_2)=\rho^{(2)}(\xx_2) \exp\left(\sum_{k=1}^\infty\frac{(-1)^k}{k!} \int_{\Lambda^k}F_{2}^{(k)}(\xx_2;\yy_k)\dyy_k\right),
\end{align}
where the family $\boldsymbol{F}_{\{x_1,x_2\}}=(F_{2}^{(k)}(\xx_2;\cdot))_{k\geq1}$ is recursively defined by, see \eqref{eq:rhonexp},
\begin{align}\label{eq:fnk}
    D_{\{x_1,x_2\}}\bro
    = \rho^{(2)}(\xx_2)
    \fexp \boldsymbol{F}_{\{x_1,x_2\}}, \quad \xx_2=\{x_1,x_2\}.
\end{align}
Later, in Remark~\ref{remref} we will show how to view an equation like \eqref{eq:fnk} as a recursive definition of $(F_{2}^{(k)}(\xx_2;\cdot))_{k\geq1}$. 
Furthermore, to obtain \eqref{eq:j2exp} from \eqref{eq:fnk} we will need to show the bounds below on the integrals in the exponential on the right-hand side of \req{j2exp}.
% \begin{align}\label{eq:fnk}
%     F_{2,k}(\xx_2,\yy_k)=\frac{\rho^{(2+k)}(\xx_2,\yy_{k })}{\rho^{(2)}(\xx_2)} - \sum_{l=2}^{k }\sum_{\pi\in\Pi_l(\{\yy_k\})}\prod_{i=1}^l F_{2, |\pi_i|}(\xx_2, {\pi_i})
% \end{align}
% where $\Pi_l(\{\yy_k\})$ is the set of partitions of $\{y_1,\dots,y_k\}$ into $l$ nonempty sets.

\section{Main result}\label{sec3}

To state the main results, three assumptions are introduced.
Let $\PP$ be a $(\mu,H)$-Gibbs measure for a Hamiltonian $H$. The discussion after \eqref{eq:j2hamiltonian} motivates and justifies the following:
\begin{assumption}\label{ass:A}
We assume for any bounded set $\Delta\subset \R^d$ there holds
\begin{align*}
    \lim_{\Lambda\nearrow \R^d}\sup_{\xx_n \in \Delta^n}\left|
    \frac{j^{(n)}_\Lambda(\xx_n)}{j^{(0)}_\Lambda}
     -
     e^{\sum_{i=1}^n \mu(x_i)-H(\xx_n)}
     \right| = 0,
\end{align*}
for $n=1,2$.
% \begin{align*} 
%        \lim_{\Lambda\nearrow \R^d}\sup_{\xx_2 \in \Delta^2}\left|\log \int_{\Gamma_{\Lambda^c}} \exp\Big(-W(\{\xx_2\}\mid \eta)\Big) \dP(\eta)
%     -
%     \log \int_{\Gamma_{\Lambda^c}}  \dP(\eta)\right| = 0.
% \end{align*}
\end{assumption}

\begin{assumption}\label{ass:B}
There are constants $M,A,D_\rho>0$ such that for any $\Lambda\subset\R^d$ any $k\geq 1$ and any $\xx_m \in \R^d$ we have
\begin{align}\label{eq:cbrho1}
    % \frac{1}{\rho^m}
    \int_{\Lambda^k}\left|\rhot{m+k} (\xx_m,\yy_k)\right|\dyy_k \leq  
    (m+k-1)!M A^m D_\rho^k \, \prod_{i=1}^m\rho^{(1)}(x_i)
\end{align}
for $m=1,2$.
% for every $n\leq m$.
% and
% \begin{align}\label{eq:cbrho2}
%     \sup_{\xx_2 \in (\R^d)^2}\frac{1}{\rho^2}\int_{\Lambda^k}\left|\rho_T^{(2+k)} (\xx_2,\yy_k)\right|\dyy_k \leq (k+1)!C D_\rho^k .
% \end{align}
\end{assumption}

\begin{remark}
Point processes satisfying Assumption \ref{ass:B} are called  \emph{strongly Brillinger mixing}, see e.g. \cite{Heinrich16}.
\end{remark}

\begin{assumption}\label{ass:C}
Assume that $\rho^{(1)}(x_1) >0 $ for all $x_1 \in \mathbb{R}$ and that there is $r >0$ and a $d(r)$ such that 
\begin{align}\label{eq:rho2lowerbound}
    {\rho^{(1)}(x_1)\rho^{(1)}(x_2)} \leq d(r)\rho^{(2)}(\xx_2)
    \qquad \text{ for all } |x_1-x_2|\geq r.
\end{align}
\end{assumption}

% \xout{This allows to treat situations where the underlying system is hard-core and hence $\rho^{(2)}(\xx_2)$ is zero for $x_1$ and $x_2$ {\xout{to near}}.} 
% {\color{blue}

% }
% \innerblock{Non-negative multi-body interaction of finite range}

% \begin{assumption}\label{ass:A}
% The Janossy densities of $\PP$ satisfy
% \begin{align}\label{eq:A}
%     \log j_\Lambda^{(1)}(x)-\log j_\Lambda^{(0)} \to \mu 
% \end{align}
% as $\Lambda\nearrow\R^d$. 
% \end{assumption}
% \begin{assumption}\label{ass:B}
% There are constants $D>0$ and $q>0$ such that the truncated correlation functions of $\PP$ satisfy 
% \begin{align}\label{eq:Rboundcluster}
%     \sup_{x\in\R^d}\int_{\Lambda^n}  \left|{\rho}_{T}^{(1+n)}(x,\yy_n) \right|\dyy_n
%     \leq n! D q^{n}, 
% \end{align}
% for every bounded set $\Lambda\subset \R^d$.     
% \end{assumption}

Now we state the main result of this paper:

\begin{theorem}\label{thm:mainthm}
Let $\PP$ be a $(\mu,H)$-Gibbs measure that satisfies a Ruelle condition and Assumptions \ref{ass:A}, \ref{ass:B} and \ref{ass:C}.
Then for any $r>0$
% If additionally there is a decreasing function $d\colon (0,\infty)\to  (0,\infty)$ such that 
% \begin{align}\label{eq:rho2lowerbound}
%     \frac{\rho^2}{\rho^{(2)}(\xx_2)} \leq d(r) \qquad \text{ for all } |x_1-x_2|\geq r
% \end{align}
% \begin{enumerate}%[label={\arabic*.}, ref=\arabic*]
%     \item[(A)] The Janossy densities of $\PP$ satisfy 
% \begin{align}
%     \log j_\Lambda^{(1)}(x)-\log j_\Lambda^{(0)} \to \mu
% \end{align}
% as $\Lambda\nearrow\R^d$. \label{itm:A}
% \item[(B)] The truncated correlation functions of $\PP$ satisfy 
% \begin{align}\label{eq:Rboundcluster}
%     \int_{\Lambda^n}  \left|{\rho}_{T}^{(1+n)}(x,\yy_n) \right|\dyy_n
%     \leq n! D q^n, 
% \end{align}
% for some $D>0$, $0<q$ and every bounded set $\Lambda\subset \R^d$, uniformly in $x\in \R^d$. \label{itm:B}
% \end{enumerate}
there exists a sufficiently small $D_\rho$, % < q_0$, where 
% \begin{align}\label{eq:qzero}
%     q_0 = \frac{1}{2(2+\zeta D)},
%     \qquad\zeta = \frac{1}{2\log 2 -1}
% \end{align}
such that \eqref{eq:omega1_bound} and \eqref{eq:rhodtildebound} hold and thus
\begin{align}\label{eq:muexpansionfinal}
    \mu(x) = \log \rho^{(1)}(x)  +  
   \sum_{k=1}^\infty\frac{(-1)^k}{k!} \int_{(\R^d)^k}
    \fod{k}(x,\yy_k) \dyy_k
\end{align}
and 
\begin{align}\label{eq:uexpansionfinal}
    H(\xx_2)  = -\log \left(
    \frac{\rho^{(2)}(\xx_2)}{\rho^{(1)}(x_1)\rho^{(1)}(x_2)}\right)
    -
    \sum_{k=1}^\infty\frac{(-1)^k}{k!} \int_{(\R^d)^k}
    \secod{k}(\xx_2,\yy_k) \dyy_k,\quad 
\end{align}
% \xout{for all} 
where  $\xx_2=\{x_1,x_2\}$ such that $ |x_1-x_2|\geq r$, and
where the families $\bom_x=(\fod{k}(x;\cdot))_{k\geq 1}$ and $\bom_{\{ x_1,x_2\}}=(\secod{k}(x_1,x_2;\cdot))_{k\geq 1}$ are recursively defined in terms of $\bro_T$ such that
\begin{align}\label{eq:rhotilde_1}
    \frac{D_{\{x\}}\bro_T}{\rho^{(1)}(x)} = \fexp \bom_x
\end{align}
and 
\begin{align}\label{eq:rhodoubletilde}
    % \frac{D_{\{x_1,x_2\}}\bro_T}{\rho^{(2)}(\xx_2)}
    % = \left(\bom_{(x_1,x_2)} -\frac{\rho^{(1)}(x_1)\rho^{(1)}(x_2)}{\rho^{(2)}(\xx_2)}\right)\ast \frac{D_{\{x_1\}}\bro_T}{\rho^{(1)}(x_1)}
    % \ast \frac{D_{\{x_2\}}\bro_T}{\rho^{(1)}(x_2)}
    \frac{D_{\{x_1,x_2\}}\bro_T}{\rho^{(2)}(\xx_2)}
    = \left(\fexp\bom_{(x_1,x_2)} -\frac{\rho^{(1)}(x_1)\rho^{(1)}(x_2)}{\rho^{(2)}(\xx_2)} \mathds{1} \right)\ast \frac{D_{\{x_1\}}\bro_T}{\rho^{(1)} (x_1) }
    \ast \frac{D_{\{x_2\}}\bro_T}{\rho^{(1)}(x_2 )}.
\end{align}
\end{theorem}

\begin{remark}\label{remref}
We can rewrite \req{rhotilde_1} and \req{rhodoubletilde} in a recursive way as:
$\fod{1}(x;y) = {\rho}^{(2)}_T(x,y)/\rho^{(1)}(x)$ and $\secod{1}(\xx_2;y)=\rho^{(2+1)}_T(\xx_2,y)/\rho^{(2)}(\xx_2)$. By isolating the highest order term in \eqref{eq:rhotilde_1} for $k\geq 2$ we obtain: 
\begin{align}\label{eq:rhotilde_2}
    \fod{k}(x;\yy_k) =
    \frac{\rho^{(1+k)}_T(x,\yy_k)}{\rho^{(1)}(x)}
    -\sum_{l=2}^k \frac{1}{l!}\sum_{\pi\in\Pi_l(\{\yy_k\})} \prod_{i=1}^l\fod{\kappa_i}(x;\pi_i),
\end{align}
where $\kappa_i$ is the cardinality of $\pi_i$.
As $\kappa_i$ is less than $k$ this relation allows to define $\fod{k}(x;\cdot)$ recursively in terms of $\fod{l}(x;\cdot)$ with $l<k$. Along the same line we obtain from \eqref{eq:rhodoubletilde}
\begin{equation}\label{eq:rhodoubletildeALT}
\begin{split}
      \secod{k}(\xx_2;\yy_k)
      &=\frac{\rho^{(2+k)}_T(\xx_2,\yy_k)}{\rho^{(2)}(\xx_2)} 
     -\frac{\rho_T^{(2)}(\xx_2)}{\rho^{(2)}(\xx_2)} 
     \sum_{\sigma\in \Pi_{\leq 2}(\{\yy_k\})}
     \frac{\rho^{(1+\alpha_1)}_T(x_1,\sigma_1)}{\rho^{(1)}(x_1)} 
     \frac{\rho^{(1+\alpha_2)}_T(x_2,\sigma_2)}{\rho^{(1)}(x_2)}\\
     &-\sum_{\sigma\in \Pi_{\leq 3}(\{\yy_k\})\atop 0<\alpha_3<k}
     \frac{\rho^{(1+\alpha_1)}_T(x_1,\sigma_1)}{\rho^{(1)}(x_1)}
     \frac{\rho^{(1+\alpha_2)}_T(x_2,\sigma_2)}{\rho^{(1)}(x_2)}
     \sum_{l=1}^{\alpha_3} \frac{1}{l!}\sum_{\pi\in\Pi_l(\sigma_3)}
     \prod_{i=1}^l \secod{\kappa_i} (\xx_2;{\pi_i}) \\
     &- \sum_{l=2}^{k} \frac{1}{l!}\sum_{\pi\in\Pi_l(\yy_k)}
     \prod_{i=1}^l \secod{\kappa_i} (\xx_2;{\pi_i}),
\end{split}
\end{equation}
where $\Pi_{\leq m}(\{\yy_k\})$ is the set of ordered partitions $\pi=(\pi_i)_{i=1,\ldots,m}$ of $\{y_1,\dots,y_k\}$ into $m$ sets where at least one of the sets $\pi_1,\dots,\pi_m$ is non-empty and $\kappa_i$ and $\alpha_l$ 
are the sizes of the parts of the subsets $\pi_i$ and $\sigma_l$, respectively. In contrast to $\Pi_{\leq m}(\{\yy_k\})$, where it is assumed that all $\pi_i$ are non-empty.
\end{remark}
\begin{remark}
In their work \cite{RNMG57} Nettleton and Green, see equation (3), they introduce the so-called \emph{n-body potentials of mean force} $(\pmf{n})_{n\geq 1}$ by the following recursive relation:
\begin{align}\label{eq:potmeanforce}
    \log \big(\rho^{(n)}(\xx_n)\big) =\pmf{n}(\xx_n)+ 
    \sum_{k=1}^{n-1} \sum_{ 1 \leq i_1 \leq  \ldots \leq i_k \leq n} \pmf{k}(x_{i_1}, \ldots , x_{i_k}) .
\end{align}
Furthermore, they introduce what one may consider as the associated $n$-body Mayer-functions 
\begin{align}\label{eq:nbodymayer}
    \mayerpmf{n}(\xx_n):= \exp\Big({\pmf{n}(\xx_n)}\Big)-1,
\end{align}
cf.~equations (14) in \cite{RNMG57}. The analogue of $F_{n}^{(k)}$ is defined in (12) of \cite{RNMG57}. In (15) of \cite{RNMG57} an expansion of the general $F_{n}^{(k)}$ is given
\begin{align}\label{eq:expNettletonGreen}
    F_{n}^{(k)}(\xx_n ;\yy_k) 
    = \rhot{k}(\yy_k) 
    +  \sum_{l=1}^n
    % \frac{1}{l!}\sum_{\pi\in\Pi_l(\yy_k)} 
    \sum_{1\leq i_1< \dots i_l\leq n}
    \prod_{j=1}^k
    % e^{\pmf{1}(y_j)}
    \rho^{(1)}(y_j)
    \sum_{C \in \mathcal{H}_{l,k}} \prod_{\nu\in E(C)} \mayerpmf{|\nu|},
\end{align}
where $\mathcal{H}_{l,k}$ is the set of connected hypergraphs on $l$ white vertices $x_{i_1}, \ldots, x_{i_l}$ and  $k$ black vertices $\yy_k$ such that there are no hyperedges containing only white vertices
% where $\mathcal{H}_{l,k}$ is a specific subset of the connected hypergraphs on $l+k$ vertices $x_{i_1}, \ldots, x_{i_l}, \yy_k$ 
and $\mayerpmf{|\nu|}$ is the $|\nu|$-body Mayer-function evaluated at the points contained in the hyperedge $\nu$. Comparing their expansion of $F_{n}^{(k)}$ with our expressions in terms of $\omega^{(l+k)}_l$, which in the general case would be  
\begin{align*}
    F_{n}^{(k)}(\xx_n ;\yy_k) 
    = \rhot{k}(\yy_k) 
    +  \sum_{l=1}^n
    % \frac{1}{l!}\sum_{\pi\in\Pi_l(\yy_k)} 
    \sum_{1\leq i_1< \dots i_l\leq n}
    \omega^{(l+k)}_l (x_{i_1},\ldots x_{i_l},\yy_k),
\end{align*}
we identify $\omega^{(l+k)}_l$ as the following term in the expansion  of Nettleton and Green
\begin{align}\label{eq:omegaclusterexpgeneral}
    \omega^{(l+k)}_l (x_{i_1}, \ldots x_{i_l},\yy_k)
    =\prod_{j=1}^k
    \rho^{(1)}(y_j)
    \sum_{C \in \mathcal{H}_{l,k}} \prod_{\nu\in E(C)} \mayerpmf{|\nu|}.
\end{align}
In Theorem~\ref{thm:mainthm} we establish a bound for $|\omega^{(l+k)}_l|$ in the case $k=1,2$, which is sufficiently strong to make the reconstruction of the potential rigorous in terms of $\omega^{(l+k)}_l$'s, see \eqref{eq:muexpansionfinal} and \eqref{eq:uexpansionfinal}. Note that in \eqref{eq:omegaclusterexpgeneral} the number of summands is finite, so \eqref{eq:expNettletonGreen} is well-defined. 
However, as is quite common in cluster expansion, this does not necessarily imply that we obtain an absolute convergence in terms of $| \mayerpmf{n}|$, that is, we cannot put the absolute value all the way through in \eqref{eq:expNettletonGreen}.
Indeed, where to put the absolute value in a cluster expansion is highly non trivial as one has to exploit possible cancelations. In the expansion of Nettleton and Green it is not immediate, which would be the natural position for the absolute value. An attempt to put it on all of the $\mayerpmf{n}$ will not work. Our result gives an answer for one such choice, namely putting the absolute value around $\omega^{(l+k)}_l$ and the cancelations are hidden  in the estimate obtained. Alternatively to our methods, one could imagine to obtain an analog of a tree-graph inequality or a completely different method, which may allow to put the absolute value through parts of the sum in $C  \in \mathcal{H}_{l,k}$. Whether this becomes necessary depends on the specific question to be considered.
\end{remark}

Recall that from \req{j2hamiltonian} and \req{j2exp} it holds
\begin{align*}
     & \mu(x_1)+\mu(x_2)-H(\xx_2) + \log \int_{\Gamma_{\Lambda^c}} \exp\Big(-W(\xx_2\mid \eta)\Big) \dP(\eta)
     \\ &= \log \rho^{(2)}(\xx_2) 
     +\sum_{k=1}^\infty
     \frac{(-1)^k}{k!} 
     \int_{\Lambda^k } F_2^{(k)}(\xx_k;\yy_k) \dyy_k.
\end{align*}
Thus, it is natural to write $( F_2^{(k)}(\xx_k;\cdot))_{k\geq 1}$ as follows:
\begin{lemma}\label{lemma:split}
For the family $(F_2^{(k)}(\xx_2; \cdot))_{k\geq 1}$ defined by \req{fnk}, it holds that
\begin{align}\label{eq:fnkdirect}
    F_{2}^{(k)}(\xx_2,\yy_{k})= \rhot{k}(\yy_k)
    +\fod{k}(x_1,\yy_k)+\fod{k}(x_2,\yy_k)
    +\secod{k}(\xx_2,\yy_k)
\end{align}
where  $(\fod{k}(x ; \cdot\} )_{k\geq 1}$ and $(\secod{k}( x_1, x_2 ; \cdot) )_{k\geq 1}$  are given by \eqref{eq:rhotilde_2} and \eqref{eq:rhodoubletildeALT}.
\end{lemma}

In \eqref{eq:fnkdirect} the term $\rho_t^{(k)}$ corresponds to $ \log \int_{\Gamma_{\Lambda^c}} \exp\Big(-W(\xx_2\mid \eta)\Big) \dP(\eta)$, the terms $\bom_x$ to $\mu$ and hence $H(\xx_2)$ to $\bom_{\{ x_1, x_2\} }$.

\begin{proof}
Recall the definition of $\BF$ in \req{fnk}. Using $\bro =\fexp\bro_T$ and \req{chainrule}, one finds
\begin{align}
    \rho^{(2)}(\xx_2) \fexp \BF 
    & =D_{\{x_1,x_2\}}\bro =
    D_{\{x_1\}}(D_{\{x_2\}}\fexp\bro_T)
    \\ & =D_{\{x_1\}}(\fexp\bro_T\ast D_{\{x_2\}}\bro_T),\nonumber
\end{align}
which by \req{prodrule} is equal to 
\begin{align*}
    %\frac{D_{\{x_1\}}(\fexp\bro_T\ast D_{\{x_2\}}\bro_T)}{\rho^{(2)}(\xx_2)}
    %&=\frac{\fexp\bro_T\ast  D_{\{x_1\}}\bro_T\ast D_{\{x_2\}}\bro_T
    %+\fexp\bro_T\ast D_{\{x_1,x_2\}}\bro_T}{\rho^{(2)}(\xx_2)} \\
    %&
    =\fexp\bro_T\ast \left( D_{\{x_1\}}\bro_T\ast D_{\{x_2\}}\bro_T
    + D_{\{x_1,x_2\}}\bro_T \right).
\end{align*}
Plugging in \req{rhodoubletilde} one obtains
\begin{align*}
 & =\fexp\bro_T\ast  D_{\{x_1\}}\bro_T\ast D_{\{x_2\}}\bro_T 
 \\ & + \fexp\bro_T\ast  \left(\rho^{(2)}(\xx_2) \fexp\bom_{(x_1,x_2)} -\rho^{(1)}(x_1)\rho^{(1)}(x_2) \mathds{1} \right)\ast \frac{D_{\{x_1\}}\bro_T}{\rho^{(1)} (x_1) }
    \ast \frac{D_{\{x_2\}}\bro_T}{\rho^{(1)}(x_2 )} 
 \\ & =  \rho^{(2)}(\xx_2) \fexp\bro_T\ast    \fexp\bom_{(x_1,x_2)}  \ast \frac{D_{\{x_1\}}\bro_T}{\rho^{(1)} (x_1) } \ast \frac{D_{\{x_2\}}\bro_T}{\rho^{(1)} (x_2) }.
\end{align*}
The result follows by \req{rhotilde_1} and \req{powerrule}.
%\begin{align*}
 %   \frac{D_{\{x_1\}}\fexp\bro_T\ast D_{\{x_2\}}\bro_T}{\rho^{(2)}(\xx_2)}
 %   &=
 %   % \fexp\bro_T\ast\frac{ D_{\{x_1\}}\bro_T\ast D_{\{x_2\}}\bro_T
 %   % + D_{\{x_1,x_2\}}\bro_T}{\rho^{(2)}(\xx_2)}
 %   \fexp\bro_T\ast \frac{D_{\{x_1\}}\bro_T\ast D_{\{x_2\}}\bro_T}{\rho^{(2)}(\xx_2)} \\
 %   &+\fexp\bro_T\ast\left(\fexp\bom_{(x_1,x_2)} -\frac{\rho^{(1)}(x_1)\rho^{(1)}(x_2)}{\rho^{(2)}(\xx_2)}\right)\ast \frac{D_{\{x_1\}}\bro_T}{\rho^{(1)} (x_1) }
 %   \ast \frac{D_{\{x_2\}}\bro_T}{\rho^{(1)}(x_2 )}
 %    % &=\fexp\bro_T\ast\frac{ D_{\{x_1\}}\bro_T\ast D_{\{x_2\}}\bro_T
  %   % + D_{\{x_1,x_2\}}\bro_T}{\rho^{(2)}(\xx_2)}{\rho^{(2)}(\xx_2)}
%\end{align*}
%which can be simplified using \req{rhotilde} to 
%\begin{align*}
 %   \frac{D_{\{x_1\}}\fexp\bro_T\ast D_{\{x_2\}}\bro_T}{\rho^{(2)}(\xx_2)}
  %  &=\fexp\bro_T\ast \fexp\bom_{(x_1,x_2)}\ast \frac{D_{\{x_1\}}\bro_T}{\rho^{(1)} (x_1) }
%    \ast \frac{D_{\{x_2\}}\bro_T}{\rho^{(1)}(x_2 )} \\
%    &=\fexp\bro_T\ast \fexp\bom_{(x_1,x_2)}
%    \ast \fexp\bom_{x_1}
%    \ast  \fexp\bom_{x_2}.
%\end{align*}
%Using \req{powerrule} then proves the claim.
\end{proof}

In light of Lemma \ref{thm:reb} we need bounds on the integrals of $(\fod{k}(x_1; \cdot))_{l\geq 1}$ and $(\secod{k}(\xx_2; \cdot))_{k\geq 1}$. The first family $(\fod{k}(x_1; \cdot))_{l\geq 1}$, which appears in the expansion of the chemical potential, was already investigated in \cite{Frommer24}, where it was shown that there exists a $\widehat{C}_1>0$ such that for all $x_1 \in \R$ it holds that
\begin{align}\label{eq:omega1_bound}
    \int_{\Lambda^k}\left|
    \fod{k}(x,\yy_k) 
    \right|\dyy_k
    \leq  k! 
    \widehat{C}_1
    \left(1+ \frac{MA}{2\log 2-1}
    \right)^k D_\rho^k .
\end{align}
In Theorem~3.1 in  \cite{Frommer24} from this then  \eqref{eq:muexpansionfinal} is concluded.
Thus, it remains to find bounds for the integrals of $(\secod{k}(\xx_2; \cdot ))_{k\geq 1}$.

\begin{proposition}\label{prop:bound}
Assume that $(\rho_T^{(m)})_m$ satisfies Assumption \ref{ass:B} and Assumption \ref{ass:C}. Then there are constants $\widehat{C}_2>0$, $\chi >0$
%=\chi(A,M)>0$ and $\theta=\theta(A,M)>0 $ 
such that
\begin{align}\label{eq:rhodtildebound}
     \int_{\Lambda^k}\left|\secod{k}(\xx_2,\yy_k)\right|\dyy_k \leq 
     k!
     \widehat{C}_2
     D_\rho^k \chi^k
%     \left(
%    \min
%    \left\lbrace\frac{-
%    \chi\pm 
%    \sqrt{\chi^2-4\cdot\theta (1-2\log2)}}{2\theta}
%    \right\rbrace\right)^{-k},
%     % \left(1-\sqrt{
%    % \frac{d(r)M^3A^4}{2\mathscr{W}_1 \left(
%    % \frac{1}{2}\exp \left[\log 2 -1+d(r)MA^2\right]d(r)MA^2
%    % \right)}}\right)^{-k}
\end{align}
for every bounded $\Lambda \subset \R^d$ and all $\xx_2 \in (\R^d)^2$.
\end{proposition}
\begin{proof}
Let $(w_k)_{k \geq 1}$ be a sequence of upper bounds for the integrals over $(\secod{k})_{k \geq 1}$, i.e.,
\begin{align}\label{eq:secodbd}
    \int_{\Lambda^k}\left|\secod{k}(\xx_2,\yy_k)\right|\dyy_k
    \leq w_k \qquad \forall \xx_2 \in (\R)^2,
\end{align}
which we want to establish inductively.
Denote by $(a_k)_{k\geq 0}$ and $(c_k)_{k\geq 0}$ sequences such that 
\begin{align}\label{eq:ak}
    % \frac{1}{\rho^{(1)}(x)}
    \frac{1}{\rho^{(1)}(x)}
    \int_{\Lambda^k}\left|\rho_T^{(1+k)}(x,\yy_k)\right|\dyy_k \leq a_k \qquad \forall x \in \R
\end{align}
and 
\begin{align}\label{eq:ck}
    \frac{1}{\rho^{(2)}(\xx_2)}\int_{\Lambda^k}\left|\rho_T^{(2+k)}(\xx_2,\yy_k)\right|\dyy_k \leq c_k \qquad \forall x \in \R^2.
\end{align}
From Assumptions \ref{ass:B} and \ref{ass:C} it follows that they can be chosen as 
\begin{align}\label{eq:akexp}
    a_0=1, \qquad a_k= k!MA D_\rho^k, \qquad k\geq 1
\end{align}
and
\begin{align}\label{eq:ckexp}
    c_k = d(r)(k+1)!MA^2D_\rho^k, \qquad k \geq 0.
\end{align}
Integrating \req{rhodoubletildeALT} with respect to $\yy_k$ and using the triangle inequality we obtain:
\begin{align}\label{eq:wkineq1}
    &\int_{\Lambda^k}\left|\secod{k}(\xx_2,\yy_k)\right|\dyy_k
    \leq 
      \int_{\Lambda^k}\left|\frac{\rho^{(2+k)}_T(\xx_2,\yy_k)}{\rho^{(2)}(\xx_2)} \right|\dyy_k 
      \\ \label{eq:wkineq2}
     &+\left|\frac{\rho_T^{(2)}(\xx_2)}{\rho^{(2)}(\xx_2)} \right|
      \int_{\Lambda^k} \sum_{\sigma\in \Pi_{\leq 2}(\{\yy_k\})}
     \left|\frac{\rho^{(1+\alpha_1)}_T(x_1,\sigma_1)}{\rho^{(1)}(x_1)} 
     \frac{\rho^{(1+\alpha_2)}_T(x_2,\sigma_2)}{\rho^{(1)}(x_2)}
     \right|\dyy_k
     \\ \label{eq:wkineq3}
     & +\!\!\!\! \int_{\Lambda^k}\!\! \sum_{\sigma\in \Pi_{\leq 3}(\{\yy_k\})\atop 0<\alpha_3<k}\!\sum_{l=1}^{\alpha_3} \frac{1}{l!} \!\!\! \!\!\sum_{\pi\in\Pi_l(\sigma_3)}
     \!\!\left|
     \frac{\rho^{(1+\alpha_1)}_T(x_1,\sigma_1)}{\rho^{(1)}(x_1)}
     \frac{\rho^{(1+\alpha_2)}_T(x_2,\sigma_2)}{\rho^{(1)}(x_2)}
     \prod_{i=1}^l \secod{\kappa_i}\! (\xx_2,{\pi_i})
     \right| \! \dyy_k \\ \label{eq:wkineq4}
     &+\int_{\Lambda^k} \sum_{l=2}^{k}\frac{1}{l!} \sum_{\pi\in\Pi_l(\yy_k)}
     \left|
     \prod_{i=1}^l \secod{\kappa_i} (\xx_2,{\pi_i})
     \right|\dyy_k.
\end{align}
Using the fact that the contribution of each integral only depends on the size of the corresponding part of the partition and not on the particular elements, the above sums can be simplified. First, note that by assumption the term on the right-hand side of \req{wkineq1} can be bounded by $c_k$. Second, for the terms in \req{wkineq2} it follows that
\begin{align*}
    &\left|\frac{\rho_T^{(2)}(\xx_2)}{\rho^{(2)}(\xx_2)} \right|
     \int_{\Lambda^k} \sum_{\sigma\in \Pi_{\leq 2}(\{\yy_k\})}
      \left|\frac{\rho^{(1+\alpha_1)}_T(x_1,\sigma_1)}{\rho^{(1)}(x_1)} 
     \frac{\rho^{(1+\alpha_2)}_T(x_2,\sigma_2)}{\rho^{(1)}(x_2)}
     \right|\dyy_k 
     \\
     &\leq 
     \left|\frac{\rho_T^{(2)}(\xx_2)}{\rho^{(2)}(\xx_2)} \right|
     \sum_{l=0}^k\binom{k}{l}
      \int_{\Lambda^l}\left|\frac{\rho^{(1+l)}_T(x_1,\yy_l)}{\rho^{(1)}(x_1)} \right|\dyy_l
      \int_{\Lambda^{k-l}}\left|
     \frac{\rho^{(1+k-l)}_T(x_2,\yy_{k-l})}{\rho^{(1)}(x_2)}
     \right|\dyy_{k-l}
     \\ & \leq 
     c_0
     \sum_{l=0}^k\binom{k}{l}
     a_la_{k-l}, 
\end{align*}
where in the last inequality we used \req{ak} and that $a_0 =1$. Lastly, the terms in \req{wkineq3} and \req{wkineq4} can be re-written using Bell polynomials. Recall that the $k$-th \emph{exponential Bell polynomial} is defined by
\begin{align}\label{eq:bellpoly}
    B_k (t_1,\dots,t_k) = 
    \sum_{l=1}^k \frac{1}{k!}
    \sum_{\pi \in \Pi_l(k)} \prod_{i=1}^l t_{\kappa_i}, \qquad B_0 =1.
\end{align}
For the terms in \req{wkineq4} one concludes
\begin{align*}
    &\sum_{l=2}^{k} \frac{1}{l!}\sum_{\pi\in\Pi_l(\yy_k)}
     \int_{\Lambda^k}\left|
     \prod_{i=1}^l \secod{\kappa_i} (\xx_2,{\pi_i})
     \right|\dyy_k \\
     &= B_k
     \left(
     \int_{\Lambda^1}\left|
     \secod{1} (\xx_2,y_1)
     \right|\dy_1, \dots,
     \int_{\Lambda^{k-1}}\left|
     \secod{k-1} (\xx_2,\yy_{k-1})
     \right|\dyy_{k-1}, 0
     \right),
\end{align*}
which can be bounded using \eqref{eq:secodbd} by
$
= B_k
     \left(w_1,\dots,w_{k-1},0
     \right).
$
Similarly, for the terms in \req{wkineq3} it follows that
\begin{align*}
    &\int_{\Lambda^k} \sum_{\sigma\in \Pi_{\leq 3}(\{\yy_k\})\atop 0<\alpha_3<k}
    \sum_{l=1}^{\alpha_3}\frac{1}{l!} \sum_{\pi\in\Pi_l(\sigma_3)}
     \left|
     \frac{\rho^{(1+\alpha_1)}_T(x_1,\sigma_1)}{\rho^{(1)}(x_1)}
     \frac{\rho^{(1+\alpha_2)}_T(x_2,\sigma_2)}{\rho^{(1)}(x_2)}
     \prod_{i=1}^l \secod{\kappa_i} (\xx_2,{\pi_i})
     \right|\dyy_k \\
     &\leq  
     \sum_{l=1}^{k-1} \binom{k}{l}
     B_l\left(w_1,\dots,w_{l}
     \right)
     \sum_{m=0}^{k-l}\binom{k-l}{m}
     a_m a_{k-l-m}.
\end{align*}
Adding the above results and using \req{secodbd} one finally arrives at
% \begin{align*}
%     &\int_{\Lambda^k}\left|\secod{k}(\xx_2,\yy_k)\right|\dyy_k
%     \leq 
%       \int_{\Lambda^k}\left|\frac{\rho^{(2+k)}_T(\xx_2,\yy_k)}{\rho^{(2)}(\xx_2)} \right|\dyy_k 
%       \\
%      &+\left|\frac{\rho_T^{(2)}(\xx_2)}{\rho^{(2)}(\xx_2)} \right|
%      \sum_{l=0}^k\binom{k}{l}
%       \int_{\Lambda^l}\left|\frac{\rho^{(1+l)}_T(x_1,\yy_l)}{\rho^{(1)}(x_1)} \right|\dyy_l
%       \int_{\Lambda^{k-l}}\left|
%      \frac{\rho^{(1+k-l)}_T(x_2,\yy_{k-l})}{\rho^{(1)}(x_2)}
%      \right|\dyy_{k-l}
%      \\
%      &+\sum_{l=1}^{k-1}
%      B_l(\int_{\Lambda}\left|\secod{1}(\xx_2,y_1)\right|\dy_1,\dots,\int_{\Lambda^l}\left|\secod{l}(\xx_2,\yy_l)\right|\dyy_l) \\
%      &\sum_{m=0}^{k-l}\binom{k-l}{m}
%      \int_{\Lambda^m}\left|\frac{\rho^{(1+m)}_T(x_1,\yy_m)}{\rho^{(1)}(x_1)} \right|\dyy_m
%       \int_{\Lambda^{k-l-m}}\left|
%      \frac{\rho^{(1+k-l-m)}_T(x_2,\yy_{k-l-m})}{\rho^{(1)}(x_2)}
%      \right|\dyy_{k-l-m}
%      % +\sum_{\sigma\in \Pi_{\leq 3}(\{\yy_k\})\atop 0<\alpha_3<k}\sum_{l=1}^{\alpha_3} \sum_{\pi\in\Pi_l(\sigma_3)}
%      % \int_{\Lambda^k}\left|
%      % \frac{\rho^{(1+\alpha_1)}_T(x_1,\sigma_1)}{\rho^{(1)}(x_1)}
%      % \frac{\rho^{(1+\alpha_2)}_T(x_2,\sigma_2)}{\rho^{(1)}(x_2)}
%      % \prod_{i=1}^l \secod{\kappa_i} (\xx_2,{\pi_i})
%      % \right|\dyy_k.
% \end{align*}
\begin{align*}
 &   \int_{\Lambda^k}\left|\secod{k}(\xx_2,\yy_k)\right|\dyy_k
    \leq \,\,
c_k + B_k\left(w_1,\dots,w_{k-1},0\right) \\
&+\frac{\rho^{(2)}_T(\xx_2)}{\rho^{(2)}(\xx_2)}\sum_{l=0}^k\binom{k}{l}
a_l a_{k-l}
+ \sum_{l=1}^{k-1} \binom{k}{l}B_l(w_1,\dots,w_l)\sum_{m=0}^{k-l}\binom{k-l}{m}
a_m a_{k-l-m}.
\end{align*}
% where the sequence $(a_k)_{k\geq 0}$ is such that $a_0=1$ and for $k\geq 1$ we have $a_k= k!C D_\rho^k$. Now by \req{rho2lowerbound} and \req{cbrho1} for $n=2$ we get
% \begin{align}
%     \frac{1}{\rho^{(2)}(\xx_2)}\int_{\Lambda^k}\left|\rho_T^{(2+k)}(\xx_2,\yy_k)\right|\dyy_k \leq d(r) (k+1)!C D_\rho^k
% \end{align}
% for $k\geq 0$. We thus define
% inally, denoting by $(c_k)_{k\geq 0}$ a sequence such that $\frac{\rho^{(2)}_T(\xx_2)}{\rho^{(2)}(\xx_2)}\leq c_0$ and for $k\geq 1$
% \begin{align}\label{eq:ckdef}
%     \frac{1}{\rho^{(2)}(\xx_2)}\int_{\Lambda^k}\left|\rho_T^{(2+k)}(\xx_2,\yy_k)\right|\dyy_k \leq c_k,
% \end{align}
% we define
Thus, defining $(w_k)_{k \geq 1} $ with the recursion $w_1:=c_1+2c_0a_1$ and
\begin{align}\label{eq:wkdef}
    w_k:= &
    c_k + B_k\left(w_1,\dots,w_{k-1},0\right) 
+c_0\sum_{l=0}^k\binom{k}{l}
a_l a_{k-l}
\\ & +\sum_{l=1}^{k-1} \binom{k}{l}B_l(w_1,\dots,w_l)\sum_{m=0}^{k-l}\binom{k-l}{m}
a_m a_{k-l-m}, \nonumber
\end{align}
for $k\geq 2$ then it satisfies the bound in \req{secodbd}. The rest of this proof is devoted to finding an asymptotic bound for the sequence $(w_k)_{k\geq 1}$ using its exponential generating function. That is, finding a lower bound in the radius of convergence of the exponential generating function will give an upper bound on the asymptotics for $(w_k)_{k\geq 1}$.
% where the sequence $(c_k)_{k\geq 0}$ is given by $c_k = d(r)(k+1)!CD_\rho^k$ as short-hand.
% Note that, since $\rho^{(2)}\to 0 $ as $x_1-x_2\to 0$, the sequence $(c_k)_{k\geq 0}$ depends on $\xx_2$. \\
Using the fact that the $k$-th Bell polynomial is linear in the last argument, i.e., that
\begin{align*}
    B_k\left(s_1,\dots,s_{k-1},s_k+s_k'\right) =
    B_k\left(s_1,\dots,s_{k-1},s_k\right)
    +B_k\left(s_1,\dots,s_{k-1},s_k'\right),
\end{align*}
one can rewrite \req{wkdef} as  
\begin{align}\label{eq:wkdef2}
   2  w_k  = & \ c_k +
    B_k\left(w_1,\dots%,w_{k-1}
    ,w_k\right) +c_0\sum_{l=0}^k\binom{k}{l}
a_l a_{k-l}
% \\ & 
+\sum_{l=1}^{k-1} \binom{k}{l}B_l(w_1,\dots,w_l)\sum_{m=0}^{k-l}\binom{k-l}{m}
a_m a_{k-l-m}  \nonumber
%\\ = & \ c_k  +(c_0 -1)\sum_{l=0}^k\binom{k}{l}
%a_l a_{k-l} +\sum_{l=0}^{k} \binom{k}{l}B_l(w_1,\dots,w_l)\sum_{m=0}^{k-l}\binom{k-l}{m}
%a_m a_{k-l-m}  . \nonumber
\end{align}
and $2w_1 = c_1 +2c_0 a_1 + w_1$.
We denote by $E_w$ the exponential generating function $E_w$ for $(w_k)_{k\geq 1}$. By \req{wkdef2} it follows that for any $K \in \mathbb{N} \cup \{ \infty\}$ holds that
\begin{align}
    &2 E^K_w(t) = 2\sum_{k=1}^K \frac{t^k}{k!} w_k \label{eq:Ewineq} \\
     \leq \,\,&\sum_{k=2}^K \frac{t^k}{k!} \left(
    c_k
+(c_0-1)\sum_{l=0}^k\binom{k}{l}
a_l a_{k-l}
+\sum_{l=0}^{k} \binom{k}{l}B_l(w_1,\dots,w_l)\sum_{m=0}^{k-l}\binom{k-l}{m}
a_m a_{k-l-m}\right) \nonumber \\ 
& +( c_1 +2c_0 a_1 + w_1)t + (c_0 + (c_0-1) a_0^2 + B_0 a_0^2) - 2c_0 , \nonumber
\end{align}
where we get an inequality because on the right hand side we add further non negative term. The last two brackets cancel, because $B_0=a_0=1$. In the case $ K= \infty$ we have equality, but both sides may be infinity.
Denoting by $E_c$ and $E_a$ the exponential generating functions of the sequences $(c_k)_{k\geq 0}$ and $(a_k)_{k\geq 0}$ respectively, it follows that
\begin{align*}
   &2 E^K_w(t) 
  % =
   %1-c_0
   %\\
   %&+ \frac{1}{2}\sum_{k=0}^\infty \frac{t^k }{k!} 
   %\left(
    %c_k
%+(c_0-1)\sum_{l=0}^k\binom{k}{l}
%a_l a_{k-l}
%+\sum_{l=0}^{k} B_l(w_1,\dots,w_l)\sum_{m=0}^{k-l}\binom{k-l}{m}
%a_m a_{k-l-m}\right) \\
%&=1-c_0 +
%\frac{1}{2}\left(
%E_c(t)+(c_0-1)E_a(t)^2+\exp\left(E_w(t)\right)E_a(t)^2
%\right),
\leq E_c(t) + (c_0-1) E_a(t)^2 + \exp\left(E^K_w(t)\right)E_a(t)^2 -2c_0,
\end{align*}
which gives:
\begin{align}\label{eq:wkfunctionaleq}
   2E^K_w(t) +2c_0-E_c(t)-(c_0-1)E_a(t)^2
=
\exp\left(E^K_w(t)+2\log E_a(t)\right).
\end{align}
%Adding and subtracting $4\log E_a(t)$ gives
%\begin{align}\label{eq:lambert1}
%     & 2\left(E_w(t)+2\log E_a(t)\right)-4\log E_a(t)+2c_0-E_c(t)-(c_0-1)E_a(t)^2
%     \\ & =
%\exp\left(E_w(t)+2\log E_a(t)\right).
%\end{align}
Using the substitutions
\begin{align*}
    s^K(t) :=  E^K_w(t) + \frac{1}{2} \left( 2c_0-E_c(t)-(c_0-1)E_a(t)^2 \right) 
\end{align*}
and 
\begin{align*}
    \ell(t) :=2\log E_a(t) -  \frac{1}{2} \left( 2c_0-E_c(t)-(c_0-1)E_a(t)^2 \right), 
\end{align*}
one can write \eqref{eq:wkfunctionaleq} as
\begin{align*}
    2s^K(t)  \leq e^{s^K(t)+\ell(t)},
\end{align*}
or
\begin{align*}
    -s^K(t) e^{-s^K(t)} \geq -\frac{1}{2}e^{\ell(t)}.
\end{align*}
Noting that $s^K$ is a continuous function with $s_K(0)= \frac{1}{2}$ and $\ell$ is a monotone increasing analytic function for non-negative $t$ near zero with $\ell(0) = -1/2$. The function is $x \mapsto -x e^{-x}$ with a unique minimum at $x=1$ of value $-1/e$ and monotone decreasing in $( - \infty, 1 ]$. Choose $t_0 >0$ smaller than the radius of convergence of $\ell$ and such that $\ell(t) \leq \log(2) -1$, that is for all $t \in [0, t_0]$ we have
\[
-s^K(t) e^{-s^K(t)} \geq -\frac{1}{2}e^{l(t)} \geq -\frac{1}{2}e^{\ell(t_0)} \geq - \frac{1}{e} .
\]
As $-s^K(0) e^{-s^K(0)} > - \frac{1}{2} > -\frac{1}{2}e^{l(t_0)} $ by continuity there exists a minimal $t_K \in (0, \infty]$ with $S_K(t_K) \in (- \infty, 1]$ and  
\[
-\frac{1}{2}e^{\ell(t_K)} \leq -s^K(t_K) e^{-s^K(t_K)} = -\frac{1}{2}e^{\ell(t_0)}.
\]
Using the monotonicity of $\ell$ we get $t_K \geq t_0$ and thus for all $ t\in [0,t_0]$ holds that $s_K(t) \leq 1$ and hence 
\begin{align*}
    s^K(t) \leq -\mathscr{W}_0\left(
    -\frac{1}{2}
    e^{\ell(t)} \right) \leq -\mathscr{W}_0\left(
    -\frac{1}{2}
    e^{\ell(t_0)}
    \right),
\end{align*}
where $\mathscr{W}_0 : [-1/e, \infty)  \rightarrow [-1,\infty)$ is the zero-branch of the \emph{Lambert W function}, which is monotone increasing function. Hence $s^K(t)$ as a monotone sequence in $K$ converges and the limit is an analytic function with radius of convergence greater or equal to $t_0$ and hence using \eqref{eq:Ewineq} for $K = \infty$ we get
\[
s(t) = -\mathscr{W}_0\left(
    -\frac{1}{2}
    e^{\ell(t)} \right).
\] 
It remains to quantify $t_0$, that is to find the smallest $t_0$ such that
\[
\ell(t_0) =2\log E_a(t_0) -  \frac{1}{2} \left( 2c_0-E_c(t_0)-(c_0-1)E_a(t_0)^2 \right) \leq \log(2) -1.
\]
%
%Re-substituting, one then finds
%\begin{align*}
%    E_w(t)+2\log E_a(t)
%    =
%    &-\mathscr{W}\left(
%    -\frac{1}{2}
%    \exp\left(
%    -\frac{-4\log E_a(t)+2c_0-E_c(t)-(c_0-1)E_a(t)^2}{2}
%    \right)
%    \right)\\
%    &-\frac{-4\log E_a(t)+2c_0-E_c(t)-(c_0-1)E_a(t)^2}{2}.
%    % &-\mathscr{W}
%    % \left(
%    % -\frac{1}{2}
%    % \exp\left(
%    % \frac{-2c_0 +E_c(t)+(c_0-1)E_a^2(t)-4 \log E_a(t)}{2}
%    % \right)
%    % \right)\\
%    % &+
%    % \frac{-2c_0 +E_c(t)+(c_0-1)E_a^2(t)-4 \log E_a(t)}{2}.
%\end{align*}
%Since $\mathscr{W}$ is only defined for values larger than $-1/e$, one gets the condition
%\begin{align*}
%    % -\frac{1}{2}
%    % \exp\left(
%    % \frac{-2c_0 +E_c(t)+(c_0-1)E_a^2(t)-4 \log E_a(t)}{2}
%    % \right)
%    -\frac{1}{2}
%    \exp\left(
%    -\frac{-4\log E_a(t)+2c_0-E_c(t)-(c_0-1)E_a(t)^2}{2}
%    \right)
%    \geq -\frac{1}{e}
%\end{align*}
%which is equivalent to
%\begin{align*}
%    % -\frac{1}{2}
%    % \exp\left(
%    % \frac{-2c_0 +E_c(t)+(c_0-1)E_a^2(t)-4 \log E_a(t)}{2}
%    % \right)
%    \exp\left(
%    -\frac{-4\log E_a(t)+2c_0-E_c(t)-(c_0-1)E_a(t)^2}{2}
%    \right)
%    \leq \frac{2}{e}
%\end{align*}
%and thus
%\begin{align}\label{eq:almostthere}
%    {4\log E_a(t)+E_c(t)+(c_0-1)E_a(t)^2}
%    \leq 2\left(c_0+\log 2 -1\right).
%\end{align}
One can easily calculate that for $t < 1/D_\rho$ 
\begin{align*}
    E_a(t)= \sum_{k=0}^\infty
    \frac{t^k}{k!} a_k
    = \frac{MA tD_\rho}{1-tD_\rho}+1
\end{align*}
and 
\begin{align*}
    E_c(t) =\sum_{k=0}^\infty \frac{t^k}{k!}c_k
    = \frac{d(r) MA^2}{(1-tD_\rho)^2}.
\end{align*}
% In particular, it follows that
% \begin{align}\label{eq:EcEa}
%     E_c(t)= \frac{d(r)}{M} 
%     \left(
%     \frac{E_a(t)-1}{tD_\rho}
%     \right)^2,
% \end{align}
{It follows that
\begin{align*}
    {4\log\left(
    \frac{MA tD_\rho}{1-tD_\rho}+1
    \right) 
    +\frac{d(r) MA^2}{(1-tD_\rho)^2}
    +(c_0-1)\left(\frac{MA tD_\rho}{1-tD_\rho}+1\right)^2}
    \leq 
    2\left(c_0+\log 2 -1\right).
\end{align*}
Since $\log (1+s) \leq s$ for $s> -1 $ one can ensure that the above holds if
\begin{align*}
    \frac{4 MA tD_\rho}{1-tD_\rho}
    +\frac{d(r) MA^2}{(1-tD_\rho)^2}
    +(c_0-1)\left(\frac{MA tD_\rho}{1-tD_\rho}+1\right)^2
    \leq 
    2\left(c_0+\log 2 -1\right).
\end{align*}
% Recalling that $c_0=d(r) MA^2$ and factoring out one finds that
Equivalently, one can write 
% \textbf{THIS SEEMS TO BE THE WRONG DIRECTION}
\begin{align*}
    &{4 MA tD_\rho}(1-tD_\rho)
    +{d(r) MA^2}
    \\
    &+(c_0-1)\left({(MA tD_\rho)^2}
    +2{MA tD_\rho}(1-tD_\rho)
    +(1-tD_\rho)^2\right) \\
    &\leq 
    2\left(c_0+\log 2 -1\right)(1-tD_\rho)^2,
\end{align*}
which can be rewritten as
\begin{align*}
    &\left(
    (c_0-1)(MA)^2-2(c_0+1)MA
    +1-c_0-2\log 2
    \right)(tD_\rho)^2
    \\
    &+2\left(
    MA(c_0+1)+c_0-1+2\log 2
    \right)tD_\rho+
    1 - c_0 +d(r)M A^2 - 2\log2 \leq 0.
\end{align*}
Using $c_0 = d(r) MA^2$ this can be simplified to
\begin{align*}
    &\left((c_0-1)(MA)^2-2(c_0+1)MA
    +1-c_0-2\log 2
    \right)D_\rho^2t^2 \\
    &+2\left(
    MA(c_0+1)+c_0-1+2\log 2
    \right)D_\rho t 
    +1-2\log2 \leq 0.
\end{align*}
Since $1-2\log2 \leq 0$ this inequality can always be solved. Using the abbreviations 
\begin{align*}
    \chi := 
    2(MA(c_0+1)+c_0-1+2\log 2)
\end{align*}
and 
\begin{align*}
    \theta :=
    (c_0-1)(MA)^2-2(c_0+1)MA
    +1-c_0-2\log 2,
\end{align*}}
one finds the sufficient condition
\begin{align}\label{eq:tbound}
    D_\rho |t| \leq 
    \min
    \left\lbrace\left|\frac{-
    \chi\pm 
    \sqrt{\chi^2-4\cdot\theta (1-2\log2)}}{2\theta}
    \right|\right\rbrace.
\end{align}
Thus, it can be concluded that 
\begin{align*}
    w_k \lesssim 
    k! D_\rho^k
    \left(
    \min
    \left\lbrace\left|\frac{-
    \chi\pm 
    \sqrt{\chi^2-4\cdot\theta (1-2\log2)}}{2\theta}
    \right|\right\rbrace\right)^{-k},
\end{align*}
which was the claim.
\end{proof}
\begin{remark}
Using exponential generating functions {as above}, one can also significantly shorten the proof of Theorem 6.2 in \cite{Frommer24}.
\end{remark}

\begin{proofthm}
Fix $r>0$, then from Assumptions \ref{ass:B} and \ref{ass:C}, from \eqref{eq:omega1_bound}, that is  Theorem 6.2 in \cite{Frommer24}, and \req{rhodtildebound} it follows that for $|x_1-x_2|\geq r$ all three $\rho_T$, $\bom_x$ and $\bom_{\{ x_1, x_2\} }$ fulfil \eqref{eq:Fbound} and hence by \eqref{eq:fnkdirect} also $\BF$ itself.
Hence for $D_\rho$ small enough using Lemma \ref{thm:reb} we get from \eqref{eq:fnk} that   
\begin{align}\label{for}
    \log j_\Lambda^{(2)}(\xx_2) = \sum_{k=1}^\infty \frac{(-1)^k}{k!} \int_{\Lambda^k} F_{2, k}(\xx_2,\yy_{k})\dyy_k ,
\end{align}
which by Lemma~\ref{lemma:split} is equal to
\begin{align*}
    \log j_\Lambda^{(2)}(\xx_2) & = \sum_{k=1}^\infty \frac{(-1)^k}{k!} \int_{\Lambda^k} \rho_T^{(k)}(\yy_k)  \dyy_k + \sum_{k=1}^\infty \frac{(-1)^k}{k!} \int_{\Lambda^k} \bom_{x_1}(\yy_k)   \dyy_k \\ & + \sum_{k=1}^\infty \frac{(-1)^k}{k!} \int_{\Lambda^k} \bom_{x_2}(\yy_k)  \dyy_k + \sum_{k=1}^\infty \frac{(-1)^k}{k!} \int_{\Lambda^k} \bom_{\{ x_1, x_2\}}(\yy_k) \dyy_k   . \nonumber
\end{align*}
Recall the by Theorem~3.1 in  \cite{Frommer24} we have \eqref{eq:muexpansionfinal} and using \eqref{eq:j0rhot} we get that
\begin{align*}
    \log j_\Lambda^{(2)}(\xx_2) = \log j_\Lambda^{(0)} + \mu(x_1)+\mu(x_2)  + \sum_{k=1}^\infty \frac{(-1)^k}{k!} \int_{\Lambda^k}  \bom_{\{ x_1, x_2\} }(\yy_k)  \dyy_k .
\end{align*}
 Then \eqref{eq:uexpansionfinal} follows by Assumptions \ref{ass:A}.
%
%\begin{align*}
%    \int_{\Lambda^k}\! |F_{2, k}(\xx_2,\yy_{k})|\dyy_k 
%    &\leq k!\widehat{C}_3\left( \lebesgue(\Lambda)+1\right)
%     D_\rho^k\cdot \\
%     &\left(2+\frac{MA}{2\log2-1}+
%     \left(
%    \min
%    \left\lbrace\frac{-
%    \chi\pm 
%    \sqrt{\chi^2-4\cdot\theta (1-2\log2)}}{2\theta}
%    \right\rbrace\right)^{-1}\right)^k.
%\end{align*}
%Thus, for $D_\rho$ small enough Lemma \ref{thm:reb} implies
%\begin{align}\label{for}
%    \log j_\Lambda^{(2)}(\xx_2) = \sum_{k=1}^\infty \frac{(-1)^k}{k!} %\int_{\Lambda^k} F_{2, k}(\xx_2,\yy_{k})\dyy_k .
%\end{align}
%Now by Assumption \ref{ass:A} it follows that $\log j_\Lambda^{(2)}(\xx_2)-\log j_\Lambda^{(0)}\to 2\mu -H(\xx_2)$ as $\Lambda \nearrow\R^d$. Following the proof of Theorem 3.1 in \cite{Frommer24}, one can identify the terms in the expansion on the right-hand side of \eqref{for} and obtain \req{muexpansionfinal} and \req{uexpansionfinal}.
\end{proofthm}
% \newpage

\section{Examples}\label{sec4}

Let us give some examples for which we can easily check Assumption~\ref{ass:A}-\ref{ass:C}. We choose the examples quite different in order to demonstrate the wide range of system we can apply our main results to.

\subsection{Pair interaction}\label{ss:pp}
Let $u \colon \R^d \to \R\cup\{+\infty\}$ be an even function such that there exist $r_0>0$ and 
decreasing positive functions $\varphi:(0,r_0)\to \R^+_0$ and 
$\psi: [0,\infty)\to \R^+$ with
\be{varphi,psi}
   \int_{0}^{r_0} r^{d-1}\varphi(r)\dr = \infty
   \qquad \text{and} \qquad
   \int_{0}^\infty r^{d-1}\psi(r)\dr <\infty\,,
\ee
such that
\be{potentials}
\begin{aligned}
   u(x) &\,\geq\, \varphi(|x|) \qquad \textrm{ for }\ 0 < |x| < r_0\,, \\[1ex]
   |u(x)| &\,\leq\, \psi(|x|) \qquad \textrm{ for }\ |x| \geq r_0\,.
\end{aligned}
\ee
Then, it is well-known, cf.~\cite{Ruelle69},\cite{Ruelle70}, that $u$ induces a stable Hamiltonian and for constant chemical potential $\mu$, i.e.~$\mu(x)\equiv \mu$, such that the associated \emph{activity} $z=e^\mu$ is sufficiently small the corresponding (unique) $(\mu,u)$-Gibbs measure $\PP$ satisfies Assumptions \ref{ass:A} and \ref{ass:B} with $D_\rho= D_\rho(\mu)$ a decreasing function in $\mu$ (respectively $z$). Lastly, note that, see~\cite{Ruelle69,Hanke18c}, for  $z$ small enough, there holds
\begin{align*}
    \rho^2\leq ce^{u(x_2-x_1)}   \rho^{(2)}(\xx_2),
\end{align*}
for some $c>0$. Thus, for $r>0$ fixed, one gets that for 
\begin{align}
   d(r) :=  c \sup_{|x_1-x_2|>r}e^{u(x_2-x_1)}   
\end{align}
$\PP$ also satisfies Assumption \ref{ass:C}. Then, Theorem \ref{thm:mainthm} holds for $\mu$ small enough and $|x_1-x_2|>r$.

\begin{remark}
The families $(\fod{k}(x; \cdot))_{k\geq 1}$ and $(\secod{k}(x_1, x_2; \cdot) )_{k\geq 1}$ can be written as a convergent series in $z$ using cluster expansion for small enough $z$. In the sense of formal power series the expansion will hold also for Gibbs measures with interactions which are not necessarily of pair type, but we stick to pair interaction for simplicity. The expansion will be quite  similar to the expansion of the truncated correlation functions $(\rho_T^{(k)})_{k\geq 1}$, which support that $(\fod{k}(x; \cdot))_{k\geq 1}$ and $(\secod{k}(x_1, x_2; \cdot) )_{k\geq 1}$ are natural objects.

Let us recall  the cluster expansion of $\rho_T^{(2+k)}$: 
\begin{align}\label{rhot}
    \rho_T^{(2+k)} (\xx_2,\yy_k)= 
    \sum_{n=0}^\infty\frac{z^{2+k+n}}{n!} 
    \sum_{C\in\mathcal{C}_{2+k+n}}
    \int_{(\R^d)^n} \prod_{\{i,j\}\in E(C)} f_{ij} \diff {\yy'}_n,
\end{align}
where $\mathcal{C}_{2+k+n}$ is the set of connected graphs with $2+k+n$ vertices, $E(C)$ the set of edges of a graph $C$ and  $f_{ij}$ is the \emph{Mayer-function} $f_{ij}:= e^{-u(x_i-x_)}-1$ evaluated at the vertices belonging to $i$ and $j$. Starting from this, we can derive the expansion of $(\fod{k})_{k\geq 1}$ and $(\secod{k})_{k\geq 1}$. We will not study convergence of these series, though that should not be to difficult. In order to describe the graphic representation it is useful to introduce three types of vertices:
\begin{itemize}
    \item red vertices with labels $x_1$ and $x_2$;
    \item white vertices with labels $y_1,\dots,y_k$;
    \item unlabeled black vertices $y'_1,\dots,y'_n$, which are integrated in \eqref{rhot}.
\end{itemize}
We consider the connected components of a graph $C\in \mathcal{C}_{2+k+n}$ after the removal of the two red vertices and the edges attached to them. There are three different types of connected components, see also in Figure \ref{fig:graphs}:
\begin{enumerate}
    \item[(i)] The component contains no white vertices  and was connected to either one (or both) $x_1$ or $x_2$;
    % \item[(ii)] The component contains no white points and is connected to both $x_1$ and $x_2$;
    \item[(ii)] The component contains at least one white vertex and was connected to either $x_1$ or $x_2$, but not to both;
    \item[(iii)] The component contains at least one white vertex and was connected to both $x_1$ and $x_2$.
\end{enumerate}
The graphs of type (i) (all graphs on the left-hand side of Figure \ref{fig:graphs}) make up all the graphs in the expansion of $\rho^{(2)}$; more precisely one collects all graphs of type (i), restores all the edges between these connected components and the red vertices labeled $x_1$ and $x_2$ and a potential edge between $x_1$ and $x_2$. Whereas the graphs of type (ii) make up the graphs in the expansion of $\fexp \bom_{x_1} $, $\fexp \bom_{x_2} $ respectively. Lastly this leaves the graphs of type (iii) for the expansion of $\fexp \bom_{(x_1,x_2)} $, more precisely, one collects all  connected components of type (iii) and restores  all the edges which connected the red vertices to these connected components, but a potential edge between $x_1$ and $x_2$ will not be included.    Thus, one can write
\begin{align}
    \fod{k} (x;\yy_k) = \sum_{n=0}^\infty\frac{z^{1+k+n}}{n!} 
    \sum_{C\in\mathcal{C}_{1\mid k+n}}
    \int_{(\R^d)^n} \prod_{\{i,j\}\in E(C)} f_{ij} \diff {\yy'}_n
\end{align}
and 
\begin{align}
    \secod{k}(\xx_2;\yy_k)=
    \sum_{n=0}^\infty\frac{z^{2+k+n}}{n!} 
    \sum_{C\in\mathcal{C}_{2\mid k+n}}
    \int_{(\R^d)^n} \prod_{\{i,j\}\in E(C)} f_{ij} \diff {\yy'}_n
\end{align}
where $\mathcal{C}_{1\mid k+n}$ is the set of all connected graphs on $1+n+k$ vertices, with one red vertex labeled $x$, $k$ white vertices labeled $x_1, \ldots x_k$ and $n$ black vertices $y'_1, \ldots , y'_n$ such that graph stays connected upon removal of the vertex labeled $x$. Whereas $\mathcal{C}_{2\mid k+n}$ is the set of connected graphs that have $2$ red , $k$ white and $n$ black vertices and there is no edge between the red vertices $x_1 $ and $x_2$, which  stay connected upon the removal of the red vertices and all attached edges. Using tools from cluster expansion one may improve the radius of convergence of the expansions in this case.
\begin{figure}
    \centering
    \begin{tikzpicture}[
    node distance=1.5cm and 1cm,
    reddot/.style={circle, fill=red, minimum size=6pt, inner sep=0pt},
    blackdot/.style={circle, fill=black, minimum size=4pt, inner sep=0pt},
    whitedot/.style={circle, draw=black, fill=white, minimum size=4pt, inner sep=0pt},
    redlabel/.style={above, text=black, font=\small},
    whitelabel/.style={right, text=black, font=\small},
    blacklabel/.style={left, text=black, font=\small}
]

% Red points centered in the picture
\node[reddot] (x1) at (0,1) {};
\node[reddot] (x2) at (0,-1) {};
% \node[redlabel] at (x1.north) {$x_1$};
% \node[redlabel] at (x2.north) {$x_2$};
\node[redlabel] at (0,1.1) {$x_1$};
\node[redlabel] at (0,-1.5) {$x_2$};

% Black Cluster 1 (Top-left) - Triangle (3 points)
\node[blackdot] (b1) at (-2,2.4) {};
\node[blackdot] (b2) at (-2.8,2) {};
\node[blackdot] (b3) at (-2,1.8) {};

% Black Cluster 2 (Bottom-left) - Square (4 points)
\node[blackdot] (b4) at (-2.4,-1.5) {};
\node[blackdot] (b5) at (-3.2,-1) {};
\node[blackdot] (b6) at (-2.4,-.5) {};
\node[blackdot] (b7) at (-3.2,-2) {};

% Black Cluster 3 (Middle-left) - Pentagon (5 points) - connected to both x1 and x2
\node[blackdot] (b8) at (-1.8,0.8) {};
\node[blackdot] (b9) at (-2.2,0.3) {};
\node[blackdot] (b10) at (-3,0) {};
\node[blackdot] (b11) at (-1.8,-0.2) {};
\node[blackdot] (b12) at (-1.2,.3) {};

% Black Cluster 4 (Far-left) - Line (2 points)
\node[blackdot] (b13) at (-2.4,1.2) {};
\node[blackdot] (b14) at (-3.2,0.6) {};

% Labels for black points
% \node[blacklabel] at (b1.west) {$b_1$};
% \node[blacklabel] at (b2.west) {$b_2$};
% \node[blacklabel] at (b3.west) {$b_3$};
% \node[blacklabel] at (b4.west) {$b_4$};
% \node[blacklabel] at (b5.west) {$b_5$};
% \node[blacklabel] at (b6.west) {$b_6$};
% \node[blacklabel] at (b7.west) {$b_7$};
% \node[blacklabel] at (b8.west) {$b_8$};
% \node[blacklabel] at (b9.west) {$b_9$};
% \node[blacklabel] at (b10.west) {$b_{10}$};
% \node[blacklabel] at (b11.west) {$b_{11}$};
% \node[blacklabel] at (b12.west) {$b_{12}$};
% \node[blacklabel] at (b13.west) {$b_{13}$};
% \node[blacklabel] at (b14.west) {$b_{14}$};

% White Cluster 1 (Top-right) - Triangle (3 points)
\node[whitedot] (w1) at (2,2.5) {};
\node[blackdot] (w2) at (2.8,2) {};
\node[blackdot] (w3) at (1.2,2) {};

% White Cluster 2 (Bottom-right) - Square (4 points)
\node[blackdot] (w4) at (2,-1.7) {};
\node[whitedot] (w5) at (2.8,-2) {};
\node[whitedot] (w6) at (1.2,-1.4) {};
%\node[whitedot] (w7) at (2,-0.8) {};

% White Cluster 3 (Middle-right) - Pentagon (5 points) - connected to both x1 and x2
\node[blackdot] (w8) at (1.5,0.8) {};
\node[whitedot] (w9) at (2.6,0.3) {};
\node[blackdot] (w10) at (0.8,0.3) {};
\node[blackdot] (w11) at (1.8,-0.2) {};
\node[blackdot] (w12) at (1.2,-0.2) {};

% White Cluster 4 (Far-right) - Line (2 points)
\node[whitedot] (w13) at (1.6,-1) {};
\node[whitedot] (w14) at (2.4,-.4) {};

% Labels for white points
% \node[whitelabel] at (w1.east) {$w_1$};
% \node[whitelabel] at (w2.east) {$w_2$};
% \node[whitelabel] at (w3.east) {$w_3$};
% \node[whitelabel] at (w4.east) {$w_4$};
% \node[whitelabel] at (w5.east) {$w_5$};
% \node[whitelabel] at (w6.east) {$w_6$};
% %\node[whitelabel] at (w7.east) {$w_7$};
% \node[whitelabel] at (w8.east) {$w_8$};
% \node[whitelabel] at (w9.east) {$w_9$};
% \node[whitelabel] at (w10.east) {$w_{10}$};
% \node[whitelabel] at (w11.east) {$w_{11}$};
% \node[whitelabel] at (w12.east) {$w_{12}$};
% \node[whitelabel] at (w13.east) {$w_{13}$};
% \node[whitelabel] at (w14.east) {$w_{14}$};

% Connections within Black Cluster 1 (Triangle)
\draw (b1) -- (b2) -- (b3) -- (b1);

% Connections within Black Cluster 2 (Square)
\draw (b4) -- (b5) -- (b7) -- (b6) -- (b4);
\draw (b5) -- (b6);

% Connections within Black Cluster 3 (Pentagon)
\draw (b8) -- (b9) -- (b11) -- (b12) ;
\draw (b10) -- (b8);
%\draw (b9) -- (b10);
\draw (b11) -- (b8);

% Connections within Black Cluster 4 (Line)
\draw (b13) -- (b14);

% Connections within White Cluster 1 (Triangle)
\draw (w1) -- (w2) -- (w3) -- (w1);

% Connections within White Cluster 2 (Square)
% \draw (w4) -- (w5) -- (w7) -- (w6) -- (w4);
\draw (w6) -- (w4);
\draw (w5) -- (w4);

% Connections within White Cluster 3 (Pentagon)
\draw (w8) -- (w9) -- (w11) -- (w12) -- (w10) -- (w8);
% \draw (w9) -- (w10);
\draw (w11) -- (w8);

% Connections within White Cluster 4 (Line)
\draw (w13) -- (w14);

% Connect red points to black clusters - straight lines, carefully positioned
% Cluster 1 (triangle): only to x1 (connect to top node)
\draw[red] (x1) -- (b3);
% Cluster 2 (square): only to x2 (connect to bottom node)
\draw[red] (x2) -- (b4);
\draw[red] (x2) -- (b6);
% Cluster 3 (pentagon): to both x1 and x2 (connect to outer nodes)
\draw[red] (x1) -- (b8);
\draw[red] (x1) -- (b12);
\draw[red] (x2) -- (b11);
\draw[red] (x2) -- (b12);
% Cluster 4 (line): only to x1 (connect to top node)
\draw[red] (x1) -- (b13);

% Connect red points to white clusters - straight lines, carefully positioned
% Cluster 1 (triangle): only to x1 (connect to top node)
\draw[red] (x1) -- (w3);
% Cluster 2 (square): only to x2 (connect to bottom node)
\draw[red] (x2) -- (w6);
% Cluster 3 (pentagon): to both x1 and x2 (connect to outer nodes)
% \draw[red] (x1) -- (w8);
\draw[red] (x1) -- (w10);
\draw[red] (x2) -- (w11);
\draw[red] (x2) -- (w12);
% Cluster 4 (line): only to x2 (connect to top node)
\draw[red] (x2) -- (w13);

% Dashed red line between the two red points
\draw[red, dashed, thick] (x1) -- (x2);

\end{tikzpicture}
    \caption{Typical graph of the expansion of $\rho^{(2+k)}$.}
    \label{fig:graphs}
\end{figure}
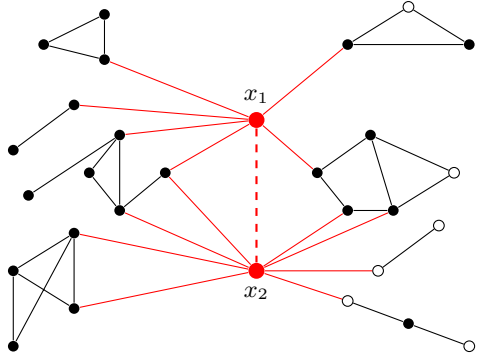

\end{remark}
\begin{remark}
As mentioned in the introduction, in computational physics the inverse Henderson problem is concerned with recovering the interaction potentials of a given system from the correlation functions. Commonly, iterative methods such as the \emph{inverse Monte-Carlo iteration (IMC)} or the \emph{inverse Boltzmann iteration} are used to approximate a solution. {The typical first guess for these iterative methods is (up to a minus sign) the $2$-body {potential of mean force}, see \req{potmeanforce}, i.e.
\begin{align*}
    u_0(\xx_2)=-\pmf{2}(\xx_2)= -\log \frac{\rho^{(2)}_*(\xx_2)}{\rho^2_*},
\end{align*}
where $\rho_*$ and $\rho^{(2)}_*$ are the density and pair correlation function of the target point process $\PP_*$. Comparison with \req{uexpansionfinal} shows that this is the first term in the expansion of $u$.
%This initial guess corresponds to the zero-order term from \req{uexpansionfinal}. 
Instead, one could add the second term of this expansion, i.e. 
\begin{align*}
    u_0(\xx_2) & = -\log \frac{\rho^{(2)}_*(\xx_2)}{\rho^2_*}
    +\int_{\R^d} \secod{1}(\xx_2,y)\dy 
    \\
    &=-\log \frac{\rho^{(2)}_*(\xx_2)}{\rho^2_*}
    +\frac{1}{\rho^{(2)}_*(\xx_2)}\int_{\R^d} \left(
    \rho^{(3)}_{T,*}(\xx_2,y)
    -\rho^{(2)}_{T,*}(\xx_2)\left(
    \frac{\rho^{(2)}_{T,*}(x_1,y)+\rho^{(2)}_{T,*}(x_2,y)}{\rho_*}\right)
    \right)\dy
\end{align*}
and use this as an initial guess for the iterative solvers.
% One could instead use the first order approximation, i.e.,
% \begin{align*}
%     u_0(\xx_2) & = -\log \frac{\rho^{(2)}_*(\xx_2)}{\rho^2_*}
%     \\ & +\frac{1}{\rho^{(2)}_*(\xx_2)}\int_{\R^d} \left(
%     \rho^{(3)}_{T,*}(\xx_2,y)
%     -\rho^{(2)}_{T,*}(\xx_2)\left(
%     \frac{\rho^{(2)}_{T,*}(x_1,y)+\rho^{(2)}_{T,*}(x_2,y)}{\rho}\right)
%     \right)\dy.
% \end{align*}
In particular, since in the IMC iteration the correlation functions up to order $4$ have to be calculated, this approach could be useful. Note that one could express $\secod{1}(\xx_2,y)$ also by using \req{omegaclusterexpgeneral}, i.e.~by using the potentials of mean force the corresponding formula would be
\begin{align*}
    \secod{1}(\xx_2,y)
    % &= \rho^3_*\left[
    % \Big(\mayerpmf{3}(\xx_2,y)+1\Big)\Big(1+\mayerpmf{2}(x_1,y)\Big)\Big(1+\mayerpmf{2}(x_2,y)\Big)
    % % +\mayerpmf{2}(x_1,y)\mayerpmf{2}(x_2,y)
    % -1
    % -\mayerpmf{2}(x_1,y)
    % -\mayerpmf{2}(x_2,y)\right]
    % \\
    % &
    =\rho^3_*\left[
    \mayerpmf{3}(\xx_2,y)\Big(1+\mayerpmf{2}(x_1,y)\Big)\Big(1+\mayerpmf{2}(x_2,y)\Big)
    +\mayerpmf{2}(x_1,y)\mayerpmf{2}(x_2,y)\right].
\end{align*}

}
% where an iterative sequence is constructed by and for $n\geq 1$ by
% \begin{align*}
%     u_{n+1} = u_n +\log\left(
%     \frac{\rho^{(2)}_k }{\rho^{(2)}_*}\right)
% \end{align*}
%  and $\rho^{(2)}_k$ is the pair correlation function of a $u_k$ Gibbs measure. 

\end{remark}

% For fixed $n$, one can rewrite \req{dpptcorr} as the sum over all cycles on the $n$ points $x_1,\dots,x_n$ and the sum of all chains over subsets of $\yy_k$ connecting the points $x_i$ and $x_j$ and get
% \begin{align}
%     \rho_T^{(n+k)}(\xx_{n},\yy_k)=z^{n+k}\sum_{\sigma\in S_{n}^{cy}}\sum_{\ww_0,\dots,\ww_n \subset \yy_k \atop 
%     \ww_i\cap\ww_j = \emptyset, \text{ }\sum |\ww_i| = k}\prod_{j=1}^{n}  
%     k(x_j-w_{j_1})
%     \prod_{l=1}^{|\ww_j|-1}k(w_l-w_{l+1})
%     k(w_{j_{|\ww_j|}}-x_{\sigma(j)})
% \end{align}
% where in case some $\ww_i$ is empty 
% \begin{align*}
%     k(x_j-w_{j_1})
%     \prod_{l=1}^{|\ww_j|-1}k(w_l-w_{l+1})
%     k(w_{j_{|\ww_j|}}-x_{\sigma(j)})
%     =k(x_j-x_{\sigma(j)}).
% \end{align*}
% Thus, it follows that 
% \begin{align*}
%     \int_{\Lambda^k}\left|
%     \rho_T^{(n+k)}(\xx_{n},\yy_k)
%     \right|\dyy_k
%     &\leq 
%     z^{n+k}
%     \sum_{\sigma\in S_{n}^{cy}}
%     \sum_{\ww_0,\dots,\ww_n \subset \yy_k \atop 
%     \ww_i\cap\ww_j = \emptyset, \text{ }\sum |\ww_i| = k}\\&\prod_{j=1}^{n} 
%     \int_{\Lambda^{|\ww_j|}}
%     |k(x_j,w_{j_1})|
%     \prod_{l=1}^{|\ww_j|-1}
%     |k(w_l,w_{l+1})|
%     |k(w_{j_{|\ww_j|}},x_{\sigma(j)})|\dww_j.
% \end{align*}
% Following the proof in \cite{Duneau75}

\subsection{The Kirkwood-closure process}
For some $ \varsigma>0$ and an even bounded function $g \colon \R^d\to [0, \infty)$, the \emph{Kirkwood closure process} is the point process $\KK$ whose correlation functions are given by
\begin{align}\label{eq:easycorrelations}
    \rho^{(n)}(\xx_n) =  
    \varsigma^n \prod_{1\leq i < j\leq n} g(x_i-x_j).
\end{align}
When $g=e^{-u}$ for some regular for which there is a $B>0 $ such that for any $n\in \mathbb{N}$ and $x,x_1,\dots,x_n \in \R^d$ there holds
\begin{align}
    \sum_{i=1}^n u(x-x_i) \geq -B,
\end{align} 
the Kirkwood closure process exists whenever $\varsigma$ is sufficiently small, \cite{Kuna07,Frommer25}
and stable pair potential as in Example \ref{ss:pp}, the Kirkwood closure process exists whenever $\varsigma$ is sufficiently small, cf.~\cite{Frommer25}. 

In this case, Assumption \ref{ass:C} is trivially satisfied as $\rho=\varsigma $ and thus
\begin{align}
    \rho^2 = e^{u}\rho^2 e^{-u} =e^{u} \rho^{(2)}(\xx_2)
\end{align}
and $d(r) :=  \sup_{|x_1-x_2|>r}e^{u(x_2-x_1)}$ similar to what has been done in Subsection~\ref{ss:pp} for Gibbs measure with  pair interactions. It is easy to see, e.g. \cite{Kuna07}, that the truncated correlation functions of the Kirkwood closure process are given by
\begin{align}\label{eq:clusterkirkwood}
    \rho^{(k)}_T(\xx_{n })= \rho^n \sum_{C\in\mathcal{C}_n}\prod_{\{ij\}\in E(C)} h(x_i-x_j),
\end{align}
where $\mathcal{C}_n$ is the set of connected graphs on $n$ vertices and $E(C)$ is the edge set of the graph $C$ and $h=g-1$ is the so-called \emph{structure function}. One can use the same techniques used for the Ursell functions of pair interaction Gibbs measures, see e.g.~\cite{Ruelle69,Procacci17}, to see that the Kirkwood closure process satisfies Assumption \ref{ass:B} for any $m\in\N$.

In order to check Assumption \ref{ass:A} one needs that the Kirkwood closure process is  a Gibbs measure, which is studied in \cite{Frommer25}.
The simple structure of the correlation functions \req{easycorrelations} allows for a more direct calculation of the functions $(\fod{k})_{k\geq 1}$ and $(\secod{k})_{k\geq 1}$.
\begin{proposition}
In the case of the Kirkwood closure process there hold
\begin{align}\label{eq:kirkwoodfod}
    \fod{k}(x,\yy_k) = %\left[\prod_{m=1}^k\Big( 1+h(x-y_m)\Big)-1\right]
    \rho^k \sum_{C\in\mathcal{C}_{1,k}}\prod_{\{ij\}\in E(C)} h_{ij}
\end{align}
and
\begin{align}\label{eq:kirkwoodsecod}
    \secod{k}(\xx_2,\yy_k) =\rho^k \sum_{C\in\mathcal{C}_{2,k}}\prod_{\{ij\}\in E(C)} h_{ij},
\end{align}
where the sets $\mathcal{C}_{1,k}$ and $\mathcal{C}_{2,k}$ are  the sets of all connected graphs with $1$ (resp. $2$) white and $k$ black vertices that contain no edges between white vertices and stay connected when the white vertices and { their} adjoining edges are removed.% See also Figure \ref{fig:thegraphsofF}.
\end{proposition}
%\begin{figure}
 %   \centering
  %  \includegraphics[width=\textwidth]{EntwickeltMitSets.pdf}
   % \caption{The graphs associated to $F_{2,1}$ and $F_{2,2}$ }
   % \label{fig:thegraphsofF}
%\end{figure}
\begin{proof}
% Straight forward computation of $(F_{2, k})_{k\geq 1}$ 
First, we will compute $(F_{2, k})_{k\geq 1}$ directly from \eqref{eq:fnk} using \req{easycorrelations} 
\begin{align*}
    \frac{\rho^{(2+k)}(\xx_2,\yy_{k })}{\rho^{(2)}(\xx_2)} &= \frac{\rho^2g(x_1-x_2)\prod_{j=1}^k\Big( 1+h(x_1-y_j)\Big)\Big(1+h(x_2-y_j)\Big)\rho^{(k)}(\yy_{k })}{\rho^2g(x_1-x_2)} \\
    &=\prod_{j=1}^k\Big( 1+h(x_1-y_j)\Big)\Big(1+h(x_2-y_j)\Big)\rho^{(k)}(\yy_{k }).
\end{align*}
Using the definition of the truncate correlation function, see e.g.\req{clusterfunctions}, it follows that
\begin{align*}
    &  \frac{\rho^{(2+k)}(\xx_2,\yy_{k })}{\rho^{(2)}(\xx_2)} 
    \\ &=\prod_{j=1}^k\Big( 1+h(x_1-y_j)\Big)\Big(1+h(x_2-y_j)\Big)\sum_{l=1}^k\frac{1}{l!}\sum_{\pi\in\Pi_l(\{\yy_k\})}\prod_{m=1}^l\rho^{(\kappa_m)}_T(\pi_m) \\
    &=
    \sum_{l=1}^k\frac{1}{l!}\sum_{\pi\in\Pi_l(\{\yy_k\})}\prod_{m=1}^l\prod_{y'\in\pi_m}^k\Big( 1+h(x_1-y')\Big)\Big(1+h(x_2-y')\Big)\rho^{(\kappa_m)}_T(\pi_m) .
\end{align*}
This shows that
the family $(F_{2, k})_{k\geq 1}$ is given by
\begin{align}\label{eq:f2kkirkwood}
    F_{2, k}(\xx_2,\yy_{k}) 
    = \prod_{j=1}^k\Big( 1+h(x_1-y_j)\Big)\Big(1+h(x_2-y_j)\Big)\rho^{(k)}_T(\yy_{k }).
\end{align}
Repeating the same calculation for $\rho^{(1+k)}(x_i,\yy_{k })/\rho^{(1)}(x_i) $ then proves the claim using  \req{fnkdirect}.

% Furthermore for the Kirkwood-closure process This lets us give an interpretation of $F_{2,k}$ in terms of graphs, for that let us think of the vertices associated to the $\yy_k$ as black and the two vertices associated to $x_1$ and $x_1$, then we can think of a multiplication with $h(x_1-y_j)$ (or $h(x_2-y_j)$) for some $j=1,\dots,k$ as adding an edge between the vertex $x_1$ (respectively $x_2$) and $y_j$ and we can write
% \begin{align*}
%     &\prod_{j=1}^k\Big( 1+h(x_1-y_j)\Big)\Big(1+h(x_2-y_j)\Big)
%     = 1
%     +\left[\prod_{j=1}^k\Big( 1+h(x_1-y_j)\Big)-1\right]
%     +\left[\prod_{j=1}^k\Big(1+h(x_2-y_j)\Big)-1\right] \\
%     &+\left[\prod_{j=1}^k\Big( 1+h(x_1-y_j)\Big)\Big(1+h(x_2-y_j)\Big)
%     - \left[\prod_{j=1}^k\Big( 1+h(x_1-y_j)\Big)-1\right]
%     -\left[\prod_{j=1}^k\Big(1+h(x_2-y_j)\Big)-1\right] -1\right]
% \end{align*}
\end{proof}
The convergence of the representations \req{kirkwoodfod} and \req{kirkwoodsecod} in the case of  the Kirkwood closure process can be shown directly using tree-graph inequality , see e.g.~\cite{Procacci17}, one can obtain a bound similar to Proposition \ref{prop:bound}.
% \begin{align}
%     \frac{1}{\rho^n}\int_{\Lambda^k}\left|
%     \rho_T^{(n+k)}(\xx_{n},\yy_k)
%     \right|\dyy_k
%     &\leq 
%     e^{B(n+k)}\sum_{C\in\mathcal{C}_k}\prod_{\{ij\}\in E(C)} h(y_i-y_j)
% \end{align}

{
\subsection{Determinantal point processes}
Let $\kappa \in L^2(\R^d)$ be a positive semidefinite even function with $\kappa (0)=1$ and $z>0$. Then there is a locally trace class operator $K \colon L^2(\R^d)\to L^2(\R^d) , f \mapsto Kf$ defined by
\begin{align*}
    (KG) (x) := \int_{\R^d} G(y)z\kappa(x-y)\dy.
\end{align*}
If $I-K$ is positive semidefinite there is a translation invariant point process $\PP_\kappa$ such that the correlation functions of $\PP$ are given by
\begin{align}\label{eq:detcorr}
    \rho^{(n)}(\xx_n) = z^n\det \left(
    \kappa(x_i-x_j)
    \right)_{1\leq i,j \leq n},
\end{align}
cf.~\cite{Soshnikov00}. In particular, the density of a determinantal point process
is given by 
\begin{align}\label{eq:dppact}
    \rho^{(1)}(x)\equiv \rho = z.
\end{align}
Since the truncated correlation functions of this point process are given by (see \cite{Soshnikov00})
\begin{align}\label{eq:dpptcorr}
    \rho_T^{(n)}(\xx_n)= (-1)^{n-1} z^n\sum_{\sigma\in S_n^{\textnormal{cy}}}\prod_{j=1}^n \kappa(x_j-x_{\sigma(j)}),
\end{align}
where $S_n^{\textnormal{cy}}$ is the set of cyclic permutations of $\{1,\dots,n\}$ one may use the techniques Duneau et al. used in \cite{Duneau76} for chain graphs to treat the integrals over cyclic permutations to show that Assumption \ref{ass:B} is satisfied if $\kappa$ is bounded.\\ 
Assumption \ref{ass:C} in this context means that one can find $d(r)$ such that 
\begin{align*}
    1 \leq d(r)(1-\kappa(x)^2)
\end{align*}
for all $|x| \geq r$.\\
Lastly, in \cite{Georgii04} Georgii and Yoo gave conditions on $K$ under which this determinantal point process is Gibbs in a general sense.
%and proved therein (Proposition 4.2 and the discussion thereafter) that under these conditions Assumption \ref{ass:A} holds. Thus, 
For $z$ sufficiently small, one should expect that equations
\req{muexpansionfinal} and \req{uexpansionfinal} hold.
}

\section*{Acknowledgement}
The authors thank J.~P.~Neumann for some suggestions to make the proof of Lemma \ref{lemma:split} much easier to read using Ruelle calculus and the referee for many insightful comments which rendered the exposition more intuitive. F.~F. acknowledges support by the Collaborative Research Center TRR~146 with corresponding funding from the DFG. D.~T. acknowledges financial support from the Italian Research Funding Agency
(MIUR) through the PRIN project ``Emergence of condensation-like phenomena in interacting particle systems: kinetic and lattice models'', grant n. 202277WX43, as well as from
Istituto Nazionale di Alta Matematica (INDAM). Furthermore, T.K.\ and D.T.\ are members of INdAM-GNAMPA.

\section*{Declaration}
Data sharing is not applicable to this article as no datasets were generated or analysed. The authors state that there are no conflicts of interest.

\end{document}